\documentclass{article}
\usepackage[preprint]{neurips_2026}
\usepackage[utf8]{inputenc}
\usepackage[T1]{fontenc}
\usepackage{amsmath,amssymb,amsfonts}
\usepackage{amsthm}
\usepackage{graphicx}
\usepackage{booktabs}
\usepackage{hyperref}
\usepackage{url}
\usepackage{microtype}
\usepackage{xcolor}
\usepackage{subcaption}
\usepackage{algorithm}
\usepackage{algorithmic}
\usepackage{multirow}
\usepackage{enumitem}
\usepackage{float}

\graphicspath{{figures/}}

\title{The Ringelmann Effect in Multi-Agent LLM Systems: A Scaling Law for Effective Team Size}

\author{%
  Bla\v{z} Bertalani\v{c}, Carolina Fortuna\\
  Jozef Stefan Institute\\
}

\newtheorem{theorem}{Theorem}
\newtheorem{proposition}{Proposition}
\newtheorem{corollary}{Corollary}
\newtheorem{definition}{Definition}

\begin{document}
\maketitle

\begin{abstract}

Inference-time multi-agent LLM scaling lacks a shared unit: counting nominal agents conflates cost with independent evidence. We derive a two-parameter scaling law $R(N) = N_\text{eff}/N = 1/(1+c(N-1)N^{-\beta})$ where the regime exponent $\beta$ classifies any configuration into one of three asymptotic regimes --- hard-ceiling at $1/c$ ($\beta = 0$), sublinear at $N^\beta/c$ ($0 < \beta < 1$), or linear ($\beta \ge 1$), and a mean-field theorem predicts that peer count $k$ and rounds $\tau$ during agent debate enter the dynamics only through their product $k\tau$. The law applies at two levels: answer diversity and correctness redundancy.

Across 44 (model $\times$ task $\times$ condition) cells spanning peer debate, self-correction, random-noise placebo, self-consistency, three open-weight families (Qwen, Llama, Ministral) at scales from 7B to 32B with a frontier API check (Gemini), thinking models, heterogeneous teams, and sparse communication, the functional form fits every condition at $R^2 > 0.99$; only $(c, \beta)$ shifts. On free-form math, dense peer influence collapses the answer-level regime from sublinear into hard-ceiling; correctness-level fits remain hard-ceiling throughout. Three findings have practical implications. \emph{(i)}~Thirty dense debating agents produce no more answer diversity than one on MMLU-Hard. \emph{(ii)}~A noise placebo tracks self-correction on free-form math and at $4\times$ scale, so within homogeneous teams the gain commonly attributed to ``debate'' comes from re-evaluation, not peer content. \emph{(iii)}~A single $N \le 5$ pilot predicts the $N=30$ structural ceiling, and within the configurations tested only architectural diversity (heterogeneous teams) lowers $c$ and escapes the hard-ceiling regime, communication-mode interventions do not.
\end{abstract}

\section{Introduction}

Multi-agent LLM systems are emerging as a promising candidate for AI automation. While pretraining scaling laws gave the field a shared unit (parameters, tokens, compute) for reasoning about when and how to scale models\citep{kaplan2020scaling,hoffmann2022training}, inference-time multi-agent scaling has no such unit. Papers report ``$N{=}10$ agents helped'' \citep{du2023improving} and ``$N{=}10$ agents didn't help'' \citep{wang2024rethinking,smit2024should,choi2025debatevote} without a shared axis for comparison. The main contribution of this paper is a \emph{measurement framework} inspired by psychosocial theory. It consists of a two-parameter scaling law $(c, \beta)$ that, estimating from a small pilot, informs how many of $N$ nominal agents actually contribute independent evidence. The law applies at two levels: \emph{answer diversity} (how distinct the team's answers are) and \emph{correctness redundancy} (how correlated the binary correct/wrong indicators are). It also has two practical consequences. First, deployment planning collapses from a three-dimensional grid search over team size, peer count, and revision rounds into a single small pilot. Second, the estimated $\beta$ tells deployers when adding agents pays off ($\beta \approx 0$ means a hard ceiling, $\beta > 0$ means scaling still adds independent evidence), and candidate interventions (architectural diversity, sparse topology, MoA pipelines, verifier filtering) can be evaluated on the same $(c, \beta)$ axis.

\textbf{Contributions.}\label{sec:intro-contributions}

\textbf{C1. (regime classifier)} \textbf{A new two-parameter scaling law $R(N){=}1/(1{+}c(N{-}1)N^{-\beta})$ that classifies an agentic  configuration into hard-ceiling ($\beta{=}0$), sublinear ($0{<}\beta{<}1$), or linear ($\beta{\ge}1$) regimes.} The law applies at both answer-diversity and correctness-redundancy levels, and parameters estimated on $N \le 5$ extrapolate to $N{=}30$ across every evaluated configuration (\S\ref{sec:framework}).

\textbf{C2. (falsifiable mechanism)} Mean-field theorem communication-density (Thm.~\ref{thm:comm-density}, \S\ref{sec:dynamics}) derives the scaling law from explicit conformity dynamics and yields the falsifiable prediction that peer count $k$ and rounds $\tau$ collapse to their product $k\tau$. Empirical confirmation within $\Delta\rho \approx 0.04$ (Table~\ref{tab:ktau_collapse}; full stratified analysis in App.~\ref{app:k_sweep}).

\textbf{C3. (unification)} The same form describes 44 configurations across debate, self-correction, noise placebo, self-consistency, three model families at two scales, a thinking model, heterogeneous teams, sparse communication, and even classical human group studies (App. B). Conditions previously treated as separate phenomena are different points in $(c, \beta)$. On free-form math, peer interaction triggers a regime crossing debate flattens to $\beta \to 0$ (hard ceiling) while self-correction stays sublinear at $\beta \approx 0.25$ (\S\ref{sec:results}).

\textbf{C4. (design-space collapse)} Once $\beta \approx 0$, which holds for instruct-class models on bounded tasks, only $\rho^{(0)}$ (initial inter-agent correlation) escapes the ceiling. Communication mode, peer count, rounds, and aggregator are interchangeable in this regime. The actionable lever is architectural diversity. The translation to accuracy goes through Condorcet majority voting, tight on MCQA and approximate on heterogeneous teams (up to ${\sim}7$~pp; \S\ref{sec:correlation}, App.~\ref{app:accuracy_projection}).

\section{Related work}
\label{sec:related}

Prior work either documents diminishing or non-monotone returns at fixed team sizes \citep{wang2024rethinking,wynn2025talk,chen2024more} or characterizes individual conformity mechanisms \citep{sharma2024towards,bellina2026conformity} without combining them into a predictive framework. What has been missing is a regime classification that operates from a small-team pilot and bears a falsifiable mechanistic prediction. We supply both.

\textbf{Multi-agent scaling.} Multi-agent debate was proposed for scalable oversight \citep{irving2018ai} and inference-time scaling \citep{du2023improving,liang2023encouraging,khan2024debating,chan2023chateval}, with \citet{li2024agentforest} claiming monotone gains from sampling-and-voting. A growing body of work reports diminishing or negative returns at fixed team sizes without a predictive framework \citep{wang2024rethinking,choi2025debatevote,yang2025revisitingmad,wu2025can,smit2024should,iclr2025mad,kaesberg2025voting,wynn2025talk}, and \citet{qian2025scaling} reported logistic saturation past $1{,}000$ agents. \citet{kim2025towards} fit a predictive model across 260 configurations on agentic benchmarks. While prior works \textit{describe}, we \textit{classify and predict} uniquely connecting saturation to inter-agent correlation. We share \citet{wynn2025talk}'s conformity diagnosis but provide a two parameter scaling law validated across three model families with $R^2>0.99$. Concurrent \citet{yang2026understanding} studies the same scaling bottleneck through diversity metrics.  Our axis is variance-matching via the Kish design effect demonstrated to be predictive: $(c, \beta)$ fitted on $N \le 5$ extrapolate to $N{=}30$ at $\le 12\%$  mean relative error across three open-weight families, beating power-law ($\ge 31\%$) and logistic ($\ge 68\%$) alternatives (\S\ref{sec:universality}, Table~\ref{tab:held_out_audit} in App.~\ref{app:cross_model}).

\textbf{Mechanistic ingredients.} Sycophancy \citep{perez2022discovering,sharma2024towards}, conformity \citep{zhu2024conformity,bellina2026conformity}, identity-weighted belief updating \citep{choi2025identity}, inter-agent sycophancy \citep{yao2025sycophancydebate,oh2025biasdebate}, and non-monotone majority-vote scaling \citep{chen2024more} have been studied as separate phenomena, and \citet{denisov2026consensus} show consensus is not verification. The Ringelmann framework absorbs these mechanisms as parameter shifts in $(c, \beta)$: sycophancy and inter-agent sycophancy raise $\rho^{(0)}$, peer conformity raises $\alpha$, identity weighting modulates $\alpha$ asymmetrically, and our two-rate decomposition (\S\ref{sec:dynamics}, App.~\ref{app:two_rate}) extends \citet{bellina2026conformity} by separating learning from sycophancy rates. The scaling law translates each parameter shift into a regime prediction.

\textbf{Baselines and classical foundations.} Self-consistency \citep{wang2023selfconsistency} is the no-interaction baseline. \citet{brown2024monkeys} establish scaling for repeated independent sampling, and \citet{snell2024scaling} for test-time compute more broadly: more compute yields more accuracy at a measurable rate. Our claim is orthogonal. Inference-time scaling describes \emph{throughput} (more samples improve accuracy under a given aggregator). The Ringelmann law describes \emph{redundancy} (more agents under interaction contribute diminishing independent evidence under a regime-classified rate). Independent sampling has $\rho = 0$ by construction, so its $N_\text{eff} = N$ and the Ringelmann curve is flat at $R(N) = 1$. Any communication introduces $\rho > 0$ and moves the system into one of three regimes (C1, Prop.~\ref{prop:ringelmann-regimes}). Our self-correction control extends \citet{madaan2023selfrefine} to isolate re-evaluation from peer influence: the gap to debate is ${\le}3$~pp while re-evaluation alone adds up to $+18$~pp. The framework builds on the Ringelmann effect \citep{ringelmann1913}, the Kish design effect \citep{kish1965survey}, and correlated-Condorcet theory \citep{ladha1992condorcet,boland1989majority}. We contribute three new results: regime classification (Prop.~\ref{prop:ringelmann-regimes}), a hard-ceiling escape condition (Prop.~\ref{prop:generic-ceiling}), and a falsifiable $k\tau$-product theorem (Thm.~\ref{thm:comm-density}).

\section{Framework: a Ringelmann scaling law for multi-agent systems}
\label{sec:framework}

\subsection{Effective team size and correlation}

LLM agents conform under interaction \citep{sharma2024towards,perez2022discovering,zhu2024conformity}, so their answers correlate rather than being independent. This is the standard statistical problem of clustered samples: when observations within a sample are correlated, the nominal sample size overstates the information content. The Kish design effect \citep{kish1965survey} provides the canonical correction, converting $N$ correlated observations into an equivalent number of independent observations whose mean has the same variance:
\begin{equation}
N_\text{eff}(N) \,=\, N/(1 + (N-1)\rho_N),
\label{eq:neff}
\end{equation}
where $\rho_N$ is the equicorrelation-scale pairwise agreement at team size $N$ (the convention is that all agent pairs are treated as equally correlated). At $\rho_N = 0$ (independent agents) $N_\text{eff}(N) = N$. As $\rho_N \to 1$ (full consensus) $N_\text{eff}(N) \to 1$, regardless of $N$. The Kish form is derived for binary correctness indicators (App.~\ref{app:variance-matching}). Applied to answer-level agreement it is a variance-matching definitional rescaling onto a common axis.

\textbf{Measuring $\rho_N$.} We use a hat to mark the empirical estimator computed from data: for a team of $N$ agents,
$
\hat{\rho}_N \,:=\, \frac{N \cdot A_\text{pair} - 1}{N-1},
$
where $A_\text{pair}$ is the probability that two randomly drawn agents give the same answer post-communication. The formula converts $A_\text{pair}$ to a correlation on the same scale Eq.~\ref{eq:neff} uses.

\textbf{From a point relation to a scaling law.} Eq.~\ref{eq:neff} converts $\rho_N$ to $N_\text{eff}$ at one team size. To characterize how $N_\text{eff}$ scales as we add agents, we need to know how $\rho_N$ depends on team size. Running the experiment at $N \in \{2, 3, 5, \ldots, 30\}$ produces an empirical sequence $\{\hat{\rho}_N\}$, to get the \emph{measured correlation profile}. We model the underlying $\rho_N$ as a power law, with classical precedent in declining-efficiency literature (App.~\ref{app:classical_ringelmann}):

\begin{definition}[Ringelmann Exponent]
We model the correlation profile as
$
\rho_N = cN^{-\beta}, \qquad c > 0,
$
and call $\beta$ the \emph{Ringelmann exponent} of the multi-agent system. We estimate $(c, \beta)$ by fitting this model to the measured profile $\{\hat{\rho}_N\}$.
\end{definition}

\noindent
Substituting the model into Eq.~\ref{eq:neff} gives the Ringelmann efficiency curve:
\begin{equation}
R(N) \,:=\, N_\text{eff}(N)/N \,=\, 1/(1 + c(N{-}1)N^{-\beta}),
\label{eq:ringelmann}
\end{equation}
with two interpretable parameters: $c$ sets the floor (asymptotic ceiling at $1/c$ when $\beta = 0$), and $\beta$ controls how fast correlation falls with team size. In this section we distinguish the measured $\hat{\rho}_N$ from the model $\rho_N$. Elsewhere in the paper, where context makes clear we are reporting a measurement, we drop the hat for brevity and write bare $\rho$ or $\rho_N$. Other subscripts indicate qualifiers (eff, ICC, mf, cond), tabulated in App.~\ref{app:effective_team_model}.

\subsection{Scaling regimes: from correlation to ceiling}

Previous section introduced the model with $\rho_N = cN^{-\beta}$ as a two-parameter family. We now classify what scaling behavior each region of the $(c, \beta)$ plane produces. The intuition: $N_\text{eff}$ grows when adding an agent contributes new independent evidence and saturates when it does not. Concretely, the asymptotic growth of $N_\text{eff}$ is governed by the product $N \cdot \rho_N$ in Eq.~\ref{eq:neff}, which counts the redundant correlation injected by the team: linear scaling ($N_\text{eff} = \Theta(N)$) requires this product to stay bounded: $\rho_N = O(1/N)$. Any constant $\rho > 0$ instead pushes $N \cdot \rho_N \to \infty$ and forces $N_\text{eff} \to 1/\rho$ (Proposition~\ref{prop:effective-scaling} in App.~\ref{sec:proposition2}). Under the power-law parameterization, the exponent $\beta$ controls which side of $1/N$ the correlation sits on, partitioning the system into three discrete regimes:

\begin{proposition}[Power-Law Ringelmann Regimes with proof in App.~\ref{app:proposition1}.]
\label{prop:ringelmann-regimes}
Within the effective-team model, suppose $\rho_N = cN^{-\beta}$ with $c>0$. Then
\[
N_\text{eff}(N) \sim
\begin{cases}
1/c, & \beta = 0 \quad \text{(\emph{hard ceiling})},\\[4pt]
N^{\beta}/c, & 0 < \beta < 1 \quad \text{(\emph{sublinear})},\\[4pt]
\Theta(N), & \beta \ge 1 \quad \text{(\emph{linear})}.
\end{cases}
\]
The linear regime gives $N_\text{eff} \to N/(1+c)$ at the boundary $\beta = 1$ and $N_\text{eff} \to N$ for $\beta > 1$.
\end{proposition}

\textbf{Reading $(c, \beta)$.} $c$ controls the floor: at $\beta{=}0$, the asymptotic ceiling is $1/c$. $\beta$ controls how $\rho_N$ shrinks with $N$: $\beta{=}0$ means correlation is constant (hard ceiling), $\beta{=}1$ means it falls like $1/N$ and $N_\text{eff}$ grows linearly, discounted by $1{+}c$, and $0{<}\beta{<}1$ is sublinear. Empirically one estimates $(c,\beta)$ from $R(N){=}N_\text{eff}(N)/N$ and asks whether an intervention lowers $c$, increases $\beta$, or both. Algorithm and worked example in App.~\ref{app:measurement}.

\subsection{Dynamic model: where does correlation come from?}
\label{sec:dynamics}

The preceding subsection takes $\rho_N$ as given and derives what it implies for scaling. We now close the loop: \emph{what determines $\rho_N$?} A simple mean-field closure, in the spirit of \citet{degroot1974reaching}'s classical consensus model, assumes that residual disagreement contracts by a constant factor per communication round. Writing $\tau$ for the number of post-communication updates and $d^{(\tau)}:=1-\rho^{(\tau)}$, we assume $d^{(\tau+1)}\approx(1-\alpha)d^{(\tau)}$. This yields the approximation:
\begin{equation}
\rho_{\text{mf}}^{(\tau)} \approx 1 - (1-\rho^{(0)})(1-\alpha)^\tau
\label{eq:rho}
\end{equation}
Here $\rho^{(0)}$ is the pre-communication correlation and $\alpha \in [0,1]$ is the fraction of remaining disagreement removed per round.

Thm.~\ref{thm:comm-density} derives this contraction under explicit conformity dynamics and finds that the effective $\alpha$ depends on peer count $k$ as $\alpha = 1 - (1-\alpha_1)^k$, where $\alpha_1$ is the per-modal-peer adoption probability. The single rate $\alpha$ averages over a learning rate $\alpha_l$ (wrong agents adopting correct modal peers) and a sycophancy rate $\alpha_s$ (correct agents adopting wrong modal peers). When $\alpha$ and $\rho^{(0)}$ do not depend on $N$, $\rho_\text{mf}$ is $N$-independent, giving $\beta = 0$ (hard ceiling). Escaping the ceiling requires both $\alpha_N \to 0$ and $\rho_{0,N} \to 0$ as $N$ grows (where $\alpha_N$ and $\rho_{0,N}$ denote $\alpha$ and $\rho^{(0)}$ at team size $N$), and neither alone suffices (Proposition~\ref{prop:generic-ceiling}).

\begin{theorem}[Mean-Field Convergence and the $k\tau$ Product Structure]
\label{thm:comm-density}
Consider $N \ge 2$ agents over discrete rounds. Coarsen each agent's answer into a binary indicator $\in \{\text{modal},\,\text{non-modal}\}$, where ``modal'' denotes the most common team answer at round $\tau$, and assume three conditions. \emph{(A1)} Each non-modal agent observes $k$ peers ($1 \le k \le N{-}1$) drawn uniformly without replacement from the $N{-}1$ other agents, adopting the modal answer with probability $\alpha_1 \in (0,1)$ per modal peer, independently. \emph{(A2)} Modal agents do not switch (an idealization, App.~\ref{app:thm-proof} shows the $k\tau$ structure is robust to small modal-switching rates). \emph{(A3)} The mode is stable across rounds.

Let $d^{(\tau)} = 1-\rho^{(\tau)}$ denote pairwise disagreement at round~$\tau$. In the high-agreement mode ($d^{(\tau)} \ll 1$):
\begin{equation}
d^{(\tau+1)} = d^{(\tau)}\,(1-\alpha_1)^{k} + O\!\bigl((d^{(\tau)})^2\bigr)
\label{eq:contraction}
\end{equation}
and iterating yields:
\begin{equation}
\rho_{\textup{mf}}(\tau,k) \approx 1 - (1-\rho^{(0)})(1-\alpha_1)^{k\tau}
\label{eq:rho-ext}
\end{equation}
with a single effective conformity rate $\alpha_1$. 
\end{theorem}
Three consequences: \emph{(i)}~$k$ and $\tau$ enter only through $k\tau$ (peer-round equivalence), \emph{(ii)}~fixed budget $B{=}k\tau$ is split-invariant to first order, \emph{(iii)}~on the pairwise agreement scale, $\rho^{(\tau)}$ is approximately $N$-independent for fixed $k \ge 1$ in the high-agreement regime, giving $\beta_\text{agr} \approx 0$. The entropy-scale exponent $\beta_\text{ent}$ (estimated against the entropy-derived $\rho_\text{eff}^{\text{ans}}$) may differ on free-form tasks where wrong-answer mass spreads across distinct strings. Proof and scope conditions in App.~\ref{app:thm-proof} (extension to placebo conditions with $\alpha_1 \to 0$, self-consistency at $\tau{=}0$, and the answer-level vs.\ pairwise-agreement scale distinction). The empirical $k\tau$ collapse and two-rate decomposition (with $\alpha_l > \alpha_s$ asymmetry) are reported in Appendices~\ref{app:k_sweep} and~\ref{app:two_rate}.

\section{Experimental setup}
\label{sec:setup}

\textbf{Models, Tasks and Team Dynamics.} The core identification models are: Qwen2.5-7B-Instruct (full grid, $N{=}1$--$30$). Cross-model studies involve: Llama-3.1-8B and Ministral-8B. Within-family $4\times$ scale: Qwen2.5-32B-Instruct. Stronger reasoning: Qwen3-8B (thinking mode) and Gemini Flash-Lite (minimal thinking). We consider the following representative multiple choice and free form tasks. \emph{MMLU-Hard}~\citep{hendrycks2021measuring}: 724 four-choice items curated from MMLU's complex-reasoning subjects (college physics, college mathematics, formal logic, econometrics, professional accounting; bounded MCQA). \emph{GSM-Hard}~\citep{gao2022pal}: 1{,}017 free-form numeric word problems (open-ended numeric answers, where wrong answers scatter across many distinct strings). \emph{GPQA Diamond}~\citep{rein2023gpqa}: 198 graduate-level four-choice science items (frontier-difficulty MCQA). We study team sizes $N \in \{1, 2, 3, 5, 7, 10, 15, 20, 30\}$, fully connected topology, three communication rounds (R1 independent, R2--R3 post-debate).
Detailed per-model scope insights in App.~\ref{sec:modelscope}.

\textbf{Conditions.} Three communication conditions per item: \emph{debate} (peers see true prior answers), \emph{self-correction} (own prior answer as ``peer''), \emph{noise} (peer answers drawn from a pre-generated pool of independent R1 outputs to other items, length-matched to debate). Self-correction and noise jointly isolate the contribution of genuine peer information: self-correction holds the prompt structure fixed while removing peer content, noise holds peer-text length and style fixed while removing topical relevance. The independent round-1 baseline ($\tau{=}0$, no communication) corresponds to self-consistency under majority voting.

\textbf{Prompt format.} All evaluations are zero-shot, with no in-context exemplars or chain-of-thought demonstrations. Agents respond in a structured RATIONALE/FINAL/CONF format (App.~\ref{sec:prompts}) held constant across conditions, enabling automated answer extraction and confidence elicitation. Holding the format fixed across conditions ensures it does not confound pairwise comparisons.

\textbf{Metrics.} For each (model, task, condition, $N$) cell we report:
(i)~\emph{team accuracy} under plurality voting at the final round, ties broken at random;
(ii)~the \emph{measured pairwise agreement} $\rho$ on the equicorrelation scale (Section~\ref{sec:framework});
(iii)~\emph{effective team size} at two levels, answer-level $N_\text{eff}^{\text{ans}} = 2^{H(\mathbf{a})}$ as primary (where $H(\mathbf{a})$ is the Shannon entropy in bits over the team's $N$ answers) and correctness-level $N_\text{eff}^{\text{corr}}$ from the marginal ICC of binary correctness indicators ($\rho_\text{ICC}^{\text{corr}}$) as complementary (Appendix~\ref{app:corr_level});
(iv)~the \emph{Ringelmann parameters} $(c, \beta)$ estimated by fitting $\rho_N = cN^{-\beta}$ to the measured profile across team sizes (algorithm in Appendix~\ref{app:measurement});
(v)~\emph{token cost}: prompt and output tokens, total and per effective agent (Appendix~\ref{app:cost}).
The two $N_\text{eff}$ levels coincide on bounded MCQA and decouple on free-form math, where distinct wrong numeric answers inflate answer entropy without lowering correctness redundancy.

\textbf{Ablations.} \emph{Communication density} ($k$-sweep): peer count $k \in \{1, 2, 4, 9, 29\}$ and rounds $\tau \in \{1, \ldots, 6\}$ at fixed $N$, evaluating the $k\tau$ product structure. \emph{Heterogeneous teams}: equal-proportion mixes (Qwen2.5-7B / Llama-3.1-8B / Ministral-8B in $1{:}1{:}1$ ratio) at $N \in \{3, 6, 9, 15, 30\}$, evaluating the block-exchangeable extension.

\textbf{Statistical methodology at a glance.} Pairwise agreement is reported on the equicorrelation scale $\rho$ (\S~\ref{sec:framework}). Confidence intervals use item-level bootstrap with at least $300$ iterations ($1{,}000$ for the held-out extrapolation, $5{,}000$ for sanity-check verification). We estimate $(c, \beta)$ via bounded least-squares on $R(N) = N_{\text{eff}}/N$, with $\beta \in [0, 1]$ ($\beta > 1$ would imply super-linear scaling, empirically irrelevant here) and $c \in [0, 20]$ (a permissive upper bound that prevents the optimizer drifting in the $N{=}1$ regime). The held-out audit estimates parameters from $N \in \{2, 3, 5\}$ and extrapolates to $N \in \{7, \ldots, 30\}$. Full protocol and significance tests in App.~\ref{app:statistics}. Other details such as generation parameters, answer extraction,  hardware, limitations etc. are in App.~\ref{app:experimental_details}.

\section{Results}
\label{sec:results}
\subsection{C1: Regime classification and small-$N$ pilots predict large-$N$}
\label{sec:regimes}
On Qwen2.5-7B across three tasks (Table~\ref{tab:regime_comparison}), three communication conditions (debate, self-correction, noise), and the self-consistency baseline, the same Ringelmann form $R(N) = 1/(1{+}c(N{-}1)N^{-\beta})$ applies, with parameter shifts that classify each configuration into a regime. Under dense debate, all three tasks sit in the hard-ceiling regime ($\beta \approx 0$, $c \in \{0.85, 0.57, 0.81\}$ on MMLU-Hard, GSM-Hard, GPQA respectively, asymptotic ceilings $1/c \approx 1.2$, $1.8$, $1.2$). On GSM-Hard, removing peer influence (self-correction, noise, self-consistency) shifts the regime to sublinear ($c \approx 0.33$--$0.43$, $\beta \approx 0.25$). MCQA tasks remain hard-ceiling across all conditions.

\textbf{Peer influence crosses an answer-level regime boundary on GSM-Hard.} Figure~\ref{fig:ringelmann} visualizes the crossing: noise and self-correction grow toward $N_\text{eff} \approx 6$ at $N{=}30$, while debate flattens at $N_\text{eff} \approx 1.8$. Self-consistency (round-1 only, no revision) also produces a sublinear fit ($c = 0.43$, $\beta = 0.25$). The crossing is specifically an answer-level phenomenon: at the correctness level, all GSM-Hard conditions are hard ceiling (App.~\ref{app:corr_level}, Table~\ref{tab:corr_level_fits}). Cross-family replication on Llama-3.1-8B and Ministral-8B confirms the crossing is not Qwen-specific (App.~\ref{app:cross_model}).

\textbf{Conformity collapses team productivity.} Agent-level transitions are net positive at every team size ($+8$~pp wrong $\to$ correct on GSM-Hard, App.~\ref{app:transitions}), yet team productivity remains at or below the matched no-peer control. Answer-diversity efficiency $R(N) = N_\text{eff}/N$ declines steeply with $N$ (${\sim}15{\times}$ lower at $N{=}30$ than at $N{=}2$), and $N_\text{eff}$ saturates near $1.8$ on GSM-Hard and ${\sim}1.2$ on MCQA (Table~\ref{tab:diversity}). Individual learning is real but does not translate into team gains because communication collapses $N_\text{eff}$ via inter-agent correlation. The conformity dynamics are visible to different degrees across tasks: on MMLU-Hard, round-1 agreement is already near-saturated ($\rho^{(0)} \approx 0.97$), so communication maintains rather than further raises correlation; on GSM-Hard, round-1 agreement starts low ($0.51$) and rises post-revision to $0.78$, making the conformity-amplification visible. 

\begin{table}[ht]
\caption{Canonical Ringelmann regime estimates per condition (Qwen2.5-7B).
Ringelmann form $R(N)=1/(1{+}c(N{-}1)N^{-\beta})$ estimated against observed $N_\text{eff}/N$ across team sizes. Brackets: 95\% bootstrap CIs (item-level resamples). Self-consistency uses round-1 answers (pre-communication) and is reported as a point estimate without bootstrap brackets. Full per-condition breakdown plus cross-family replication on Llama-3.1-8B and Ministral-8B in App.~\ref{app:cross_model}, Table~\ref{tab:master_fits}.}
\label{tab:regime_comparison}
\centering
\small
\begin{tabular}{llcccc}
\toprule
Task & Condition & $c$ [95\% CI] & $\beta$ [95\% CI] & Regime & $R^2$ \\
\midrule
\multirow{4}{*}{MMLU-Hard} & Self-consist. & .93 & .000 & hard ceiling & 1.000 \\
                            & Debate & .85 [.83, .86] & .000 [.00, .00] & hard ceiling & .9995 \\
                            & Self-corr. & .84 [.81, .88] & .011 [.00, .03] & hard ceiling & .9999 \\
                            & Noise & .65 [.64, .67] & .000 [.00, .00] & hard ceiling & .9991 \\
\midrule
\multirow{4}{*}{GSM-Hard}  & Self-consist. & .43 & .249 & sublinear & .9997 \\
                            & Debate & .57 [.55, .61] & .014 [.00, .04] & hard ceiling & .9989 \\
                            & Self-corr. & .34 [.31, .36] & .252 [.22, .29] & sublinear & .9995 \\
                            & Noise & .33 [.31, .36] & .272 [.24, .30] & sublinear & .9994 \\
\midrule
\multirow{4}{*}{GPQA}      & Self-consist. & .86 & .000 & hard ceiling & .9999 \\
                            & Debate & .81 [.77, .85] & .000 [.00, .01] & hard ceiling & .9993 \\
                            & Self-corr. & .78 [.73, .85] & .006 [.00, .04] & hard ceiling & 1.000 \\
                            & Noise & .53 [.49, .59] & .000 [.00, .05] & hard ceiling & .9984 \\
\bottomrule
\end{tabular}
\end{table}

\begin{figure}[t]
\centering
\includegraphics[width=\textwidth]{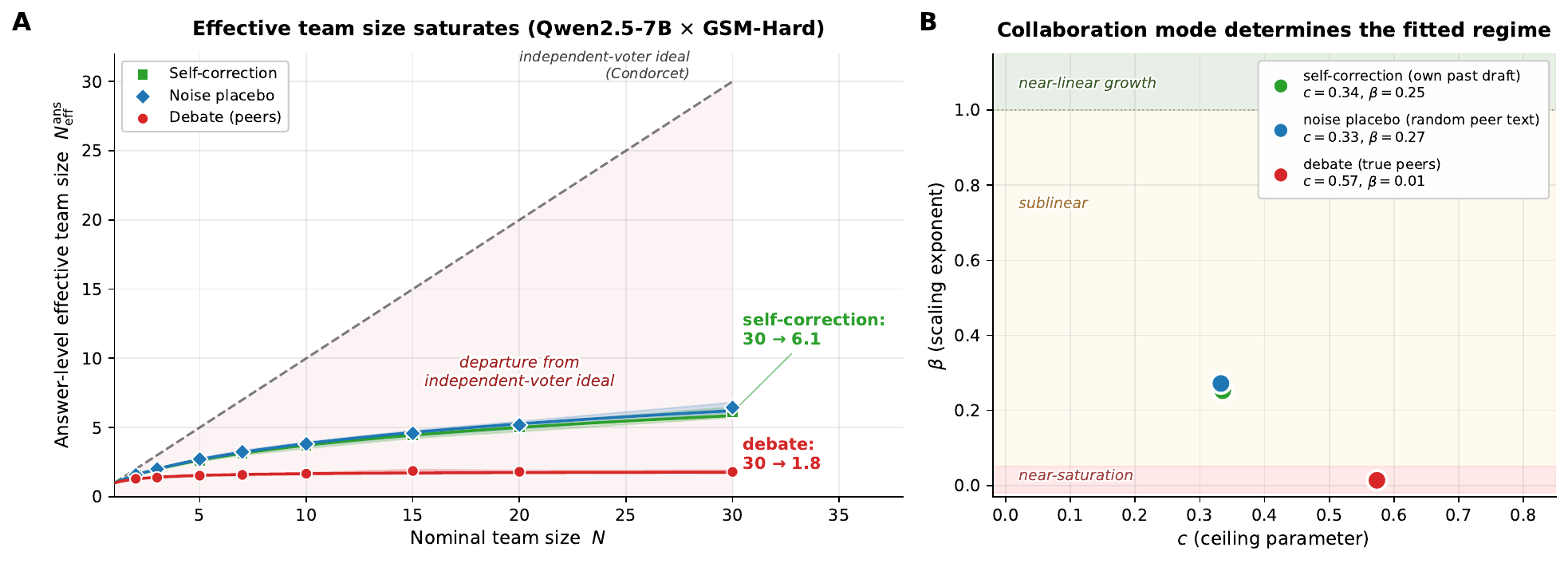}
\caption{\textbf{Effective team size saturates and the regime is set by collaboration mode}
(Qwen2.5-7B $\times$ GSM-Hard). \textbf{A:} Answer-level effective team size $N_\text{eff}^{\text{ans}}$ vs.\ nominal $N$. Self-correction and noise placebo grow sublinearly toward $N_\text{eff} \approx 6$ at $N{=}30$. Debate flattens at $N_\text{eff} \approx 1.8$. The dashed line marks the independent-voter ideal $N_\text{eff} = N$, and the shaded gap is the departure of the observed team from that ideal. \textbf{B:} Fitted Ringelmann parameters $(c, \beta)$ for the three modes. Self-correction and noise sit in the sublinear band ($\beta \approx 0.25$), debate sits in the hard-ceiling band ($\beta \approx 0$).}
\label{fig:ringelmann}
\end{figure}
\textbf{Held-out extrapolation audit.} We estimate parameters on $N \in \{2,3,5\}$ and predict $N \in \{7,10,15,20,30\}$ across 9 task$\times$condition cells. The Ringelmann form attains $5.5\%$ mean relative error, against $32.2\%$ for a generic power law and $73.4\%$ for logistic saturation (full audit table and extrapolation figure in App.~\ref{app:cross_model}, Table~\ref{tab:held_out_audit} and Figure~\ref{fig:efficiency_extrapolation}). The audit replicates on the other two open-weight families: Llama-3.1-8B preserves the ranking (Ringelmann $11.9\%$, Const-$\rho$ Kish $14.1\%$, power law $40.8\%$, logistic $69.7\%$), and Ministral-8B sharpens it (Ringelmann $5.8\%$, Const-$\rho$ Kish $6.6\%$, power law $31.6\%$, logistic $68.8\%$, on $8$ of $9$ cells, App.~\ref{app:cross_model}). Across the three open-weight families the Ringelmann mean relative error stays at $5$--$12\%$ while logistic saturation never beats $74\%$. The Const-$\rho$ Kish form (the $\beta{=}0$ special case) is competitive on hard-ceiling cells but fails on sublinear cells with $\beta > 0.2$, so the extra parameter earns its keep precisely where regime structure demands it. As an external sanity check, the estimated $c$ rank-orders the same as an independently measured correctness-level ICC across the 9 task–condition cells (full comparison in App.~\ref{app:corr_level}, Table~\ref{tab:c_vs_rho}), so $c$ tracks a measurable quantity rather than absorbing residual variance.

\subsection{C2: Peer count $k$ and rounds $\tau$ collapse to their product $k\tau$}
\label{sec:ktau}

\begin{table}[t]
\caption{$k\tau$ collapse test: pairwise agreement $\rho$ for matched $k\tau$ products at fixed $N$ (Qwen2.5-7B). For each $k\tau$, multiple $(k,\tau)$ configurations yield $\rho$ in a narrow band, confirming the theorem's prediction. Spread is the max range across configurations sharing the same product. Small $\Delta\rho$ confirms peer-round equivalence.}
\label{tab:ktau_collapse}
\centering
\small
\begin{tabular}{clccc}
\toprule
$k\tau$ & $(k,\tau)$ configs & $\rho$ at $N{=}10$ & $\rho$ at $N{=}30$ & max $\Delta\rho$ \\
\midrule
\multicolumn{5}{l}{\emph{GSM-Hard (free-form)}} \\
2  & (1,2), (2,1)               & .484--.537 & .455--.531 & $\le .076$ \\
4  & (1,4), (2,2), (4,1)        & .509--.557 & .476--.534 & $\le .058$ \\
6  & (1,6), (2,3)               & .561--.610 & .520--.548 & $\le .049$ \\
8  & (2,4), (4,2)               & .555--.610 & .515--.571 & $\le .056$ \\
12 & (2,6), (4,3)               & .659--.690 & .592--.631 & $\le .039$ \\
\midrule
\multicolumn{5}{l}{\emph{MMLU-Hard (MCQA, $K{=}4$)}} \\
2  & (1,2), (2,1)               & .889--.970 & .878--.974 & $\le .096$ \\
4  & (1,4), (2,2), (4,1)        & .857--.974 & .841--.974 & $\le .133$ \\
6  & (1,6), (2,3)               & .855--.880 & .840--.874 & $\le .034$ \\
8  & (2,4), (4,2)               & .877--.905 & .862--.888 & $\le .028$ \\
12 & (2,6), (4,3)               & .887--.896 & .853--.868 & $\le .015$ \\
\bottomrule
\multicolumn{5}{l}{\footnotesize Spread at low $k\tau$ is dominated by single-round configurations $(k,1)$.} \\  
\multicolumn{5}{l}{\footnotesize Restricting to $\tau \ge 2$ tightens the collapse (App.~\ref{app:k_sweep}, Fig.~\ref{fig:ktau_stratified}).} \\
\end{tabular}
\end{table}

The mean-field theorem (Thm.~\ref{thm:comm-density}) predicts that peer count $k$ and rounds $\tau$ enter the dynamics only through their product $k\tau$. We test this at fixed $N$ on Qwen2.5-7B with $k \in \{1, 2, 4, 9, 29\}$ and $\tau \in \{1, \ldots, 6\}$. Configurations sharing the same $k\tau$ yield agreement values within $\Delta\rho \approx 0.04$ in the high-agreement regime (Table~\ref{tab:ktau_collapse}), with looser collapse at $\tau{=}1$ where the theorem's binary linearization breaks down. The looseness at low $k\tau$ is dominated by single-round configurations $(k,1)$: on MMLU at $N{=}10$, $(4,1)$ reaches $\rho \approx 0.97$ while multi-round low-$k$ configurations sit at $\rho \approx 0.86$--$0.89$. Restricting to $\tau \ge 2$ tightens the collapse on both tasks; a per-item stratification by initial agreement (median $\Delta\rho = 0.027$ in high-agreement items vs $0.079$ in low-agreement items) is in App.~\ref{app:k_sweep}, Fig.~\ref{fig:ktau_stratified}. On the pairwise-agreement scale (the theorem's native variable) the regime-classifying exponent $\beta_\text{agr}$ decreases monotonically with $k$ toward the predicted $\beta_\text{agr} \approx 0$, and is small on both tasks ($\beta_\text{agr} \in [0.02, 0.06]$ on MMLU-Hard, $\in [0.09, 0.17]$ on GSM-Hard; Table~\ref{tab:k_sweep}). The same product structure extends under the asymmetric-rate model of App.~\ref{app:two_rate} (Table~\ref{tab:two_rate}; $\alpha_\text{eff}$ replaces $\alpha_1$, leaving $k\tau$ unchanged).

\subsection{C3: The Ringelmann form fits 44 mechanistically distinct conditions}
\label{sec:universality}

Across 44 (model $\times$ task $\times$ condition) configurations, the Ringelmann form $R(N){=}1/(1{+}c(N{-}1)N^{-\beta})$ fits with $R^2 \ge 0.97$ in-sample (median $0.999$); only the estimated $(c, \beta)$ shifts across conditions, not the form itself. We anchor on Qwen2.5-7B as the canonical example, then evaluate boundaries spanning model scale, family, paradigm, and team composition (Table~\ref{tab:master_fits}, Figure~\ref{fig:universality}; full per-cell estimates in App.~\ref{app:cross_model}; same form fits classical human Ringelmann studies in App.~\ref{app:classical_ringelmann}).

\textbf{Canonical case (Qwen2.5-7B).} As per Table~\ref{tab:master_fits}, on Qwen2.5-7B across three tasks, three communication conditions (debate, self-correction, noise), and the self-consistency baseline, the same form applies with parameter shifts that classify each configuration into a regime. Under dense debate, all three tasks sit in the hard-ceiling regime ($\beta \approx 0$, $c \in \{0.85, 0.57, 0.81\}$ on MMLU-Hard, GSM-Hard, GPQA respectively, asymptotic ceilings $1/c \approx 1.2$, $1.8$, $1.2$). On GSM-Hard, removing peer influence shifts the regime to sublinear ($c \approx 0.33$--$0.43$, $\beta \approx 0.25$). MCQA tasks remain hard-ceiling across all conditions. 

\textbf{Cross-family on MCQA: same regime, shifted ceiling.} Llama-3.1-8B and Ministral-8B both produce $\beta = 0$ on MMLU-Hard with $c \in [0.62, 0.71]$ (Qwen $c = 0.85$): the regime classification is shared, only the ceiling level differs. The same hard-ceiling pattern replicates on GPQA across all three families (Llama $c{=}0.55$, Ministral $c{=}0.57$, Qwen $c{=}0.81$, all $\beta=0$) Detailed analysis of cross-family in App.~\ref{app:noise_detail} (Table~\ref{tab:master_fits}).

\textbf{Model scale tightens the ceiling on all three tasks.} Qwen2.5-32B has $c = 0.96$ on MMLU-Hard, $c = 0.77$ on GSM-Hard, and $c = 0.97$ on GPQA (all $\beta \approx 0$), versus $c = 0.85$, $c = 0.57$, and $c = 0.81$ respectively at 7B. Noise $\approx$ self-correction replicates at 32B. Detailed analysis in App.~\ref{app:noise_detail} (Table~\ref{tab:master_fits})

\textbf{Thinking models and frontier API: tightest ceiling.} Qwen3-8B in thinking mode achieves $87.2\%$ solo accuracy on MMLU-Hard, yet debate drives $\rho = 0.983$ at $N{=}5$ ($\widehat{N}_\text{eff} = 1.02$, $c = 1.00$, $\beta \approx 0$). Gemini~3.1~Flash-Lite (minimal thinking) lands in the same regime on both GSM-Hard ($c = 0.95$) and GPQA ($c = 0.86$). As predicted, $\rho^{(0)} \to 1$ collapses $N_\text{eff}$ toward $1$ across both scale and paradigm. Detailed thinking-model analysis in App.~\ref{app:thinking} (Table~\ref{tab:thinking}).

\subsection{C4: Initial inter-agent correlation, not communication mode, controls the ceiling}
\label{sec:correlation}

Three results converge to identify the pre-communication correlation $\rho^{(0)}$ as the actionable axis.

\textbf{(i) Noise placebo $\approx$ self-correction.} On GSM-Hard, the noise placebo is indistinguishable from self-correction ($c = 0.33$ vs $c = 0.34$, $\beta = 0.27$ vs $\beta = 0.25$, $\widehat{N}_\text{eff}$ within $0.4$ at every $N$, Table~\ref{tab:regime_comparison}): once revision is held fixed, peer content per se contributes nothing to the scaling regime. The same equivalence replicates at $4\times$ scale on Qwen2.5-32B ($c = 0.45$ for both, $\beta = 0.18$ for both). Detailed analysis in App.~\ref{app:noise_detail}; the condition ordering survives item-difficulty conditioning (Table~\ref{tab:icc_decomp}).

\textbf{(ii) Mode-invariant team accuracy.} Across homogeneous conditions, communication mode shifts $(c, \beta)$ across the plane (Table~\ref{tab:regime_comparison}) yet team accuracy is invariant. A two-one-sided-tests (TOST) equivalence test, which jointly rejects ``debate is more than $5$~pp better'' and ``self-correction is more than $5$~pp better'', bounds the gap within $\pm 5$~pp on all three tasks ($p_\text{TOST} < 0.02$). Full statistics are in App.~\ref{app:corr_level} (Table~\ref{tab:decomposition}); aggregator-robustness checks in App.~\ref{app:aggregator}. The scaling law explains the convergence: on MCQA both conditions are hard-ceiling with similar $c$, and on GSM-Hard the different regimes produce similar $N_\text{eff}$ at large $N$.

\textbf{(iii) Architectural diversity raises the ceiling.} Cross-family teams (Qwen+Llama+Ministral) lower $c$ from $0.85$ to $0.54$ on MMLU-Hard, $0.57$ to $0.35$ on GSM-Hard, and $0.81$ to $0.56$ on GPQA, with $\beta \approx 0$ preserved on all three tasks. The shifts match the framework's prediction (Prop.~\ref{prop:hetero}, derived in App.~\ref{app:hetero_theory}: lower $\rho^{(0)}$ via architectural diversity raises $1/c$). Block-exchangeable validation: within-family pairs agree more than cross-family pairs in every combination (App.~\ref{app:hetero}, Table~\ref{tab:within_between_rho}, Fig.~\ref{fig:hetero_extrapolation}).

Taken together, these results pin $\rho^{(0)}$ as the only validated lever that escapes the hard-ceiling regime once $\beta \approx 0$. Communication-mode interventions (debate vs self-correction vs noise) shift parameters within the plane but leave team accuracy and the underlying ceiling structure intact.

\subsection{Deployment implications}
\label{sec:deployment}

\textbf{Where the law binds.} Multi-agent debate runs at production volume on instruct-class models ($\sim\$10^{-3}$ per problem vs.\ $\sim\$1$ on frontier reasoning models; full cost arithmetic in App.~\ref{app:deployment_econ}), so the law's predictions apply directly to the deployed regime. For frontier-class configurations with high solo accuracy, the framework predicts $c \to 1$ (\S\ref{sec:dynamics}). Qwen3-8B in thinking mode realizes this at $87.2\%$ solo accuracy on MMLU-Hard ($c = 1.00$), and Gemini 3.1 Pro's solo numbers sit in the same saturated regime (App.~\ref{app:cross_model}).

\textbf{Sizing and mode.} Once $N_\text{eff}$ saturates, additional agents spend tokens on redundant computation. On strong-model tasks ($p > 0.5$), $3$--$5$ self-correcting agents capture the full ceiling. On harder tasks, the practical knee is $N \approx 10$ (Table~\ref{tab:scaling}). Debate adds no significant accuracy benefit over self-correction while costing $10$--$45\times$ more in prompt tokens per effective agent (Table~\ref{tab:scaling}).

\textbf{Pilot procedure.} A small pilot at $N \in \{2, 3, 5\}$ recovers $(c, \beta)$ via bounded least-squares on $R(N) = N_\text{eff}/N$. This extrapolates to $N{=}30$ within $\le 12\%$ mean relative error across the evaluated families (full recipe in App.~\ref{app:measurement}; per-task projections in Table~\ref{tab:accuracy_prediction}).

\textbf{Solo accuracy moderates the accuracy returns from teaming.} On strong tasks ($p > 0.5$, MMLU-Hard), correlation prevents majority-vote gain from materializing. On harder tasks ($p \approx 0.10$--$0.30$, GSM-Hard, GPQA), re-evaluation still yields large absolute gains: GSM accuracy rises from $10\%$ to $26$--$28\%$, with self-correction ($+18.2$ pp) and debate ($+16.4$ pp) within $2$~pp of each other. Wrong-answer scatter lets the correct answer win plurality even at $p \ll 0.5$ (\S\ref{sec:framework}).

\subsection{Discussion}
\label{sec:discussion}

\textbf{What this changes.} The $(c, \beta)$ form fits 44 conditions across three open-weight families, a $4\times$ scale, a thinking model, a frontier API, and heterogeneous teams, with $5$--$12\%$ held-out error (Figure~\ref{fig:universality}, Table~\ref{tab:held_out_audit}). Protocols documented as orthogonal interventions in prior work \citep{wang2024rethinking,wynn2025talk,smit2024should,choi2025debatevote,yang2025revisitingmad} reduce to scale-and-shift transformations of one form. Concurrent diversity work \citep{yang2026understanding} reaches a similar diagnosis through entropy; we add the regime classifier and small-$N$ predictor.

\textbf{Field implication: a 1D plane.} In the regime where multi-agent debate actually deploys, communication mode is accuracy-invariant and peer content is inert; below-threshold gains ($10\% \to 26$--$28\%$ on GSM-Hard) are re-evaluation effects, as the noise-placebo result confirms. A large fraction of the multi-agent literature, varying prompts, voting rules, debate scripts, or topology while holding the model fixed, explores a plane the ceiling does not respond to. The ceiling moves only with $\rho^{(0)}$, lowered through architectural diversity (App.~\ref{app:hetero}).

\textbf{What would falsify or extend.} If verifier-aggregator pipelines lower $c$ at constant pre-verifier $\rho^{(0)}$, $\rho^{(0)}$ is not the lever. If MoA \citep{wang2024moa} or frontier heterogeneous teams produce $\beta \to 1$, the regime structure breaks. Open-ended generation extends the framework: a HumanEval pilot (App.~\ref{app:humaneval}) finds the same form on both pass-vector and canonicalized-code axes at $R^2 \ge 0.99$.

\section{Conclusion}

Across the role-symmetric LLM teams we evaluate, effective team size follows a re-parameterized Ringelmann style scaling law $R(N) = 1/(1+c(N-1)N^{-\beta})$ at both answer-diversity and correctness-redundancy levels. Systems differ not in \emph{whether} the law applies but in \emph{where} they fall in the $(c, \beta)$ landscape. On free-form math, peer influence crosses an answer-level regime boundary, collapsing sublinear scaling into a hard ceiling. Two practical implications follow: a single $N \le 5$ pilot identifies the regime, and once $\beta \approx 0$ the only validated lever that escapes the ceiling is reducing initial inter-agent correlation $\rho^{(0)}$ via architectural diversity (heterogeneous ensembles).

\bibliographystyle{plainnat}
\bibliography{references}

\clearpage
\appendix

\noindent\textbf{Appendix roadmap.}
The appendix maps to the four contributions enumerated in \S\ref{sec:intro-contributions}.

\emph{Foundations.} App.~\ref{app:effective_team_model} derives the effective-team model and proves all propositions and Thm.~\ref{thm:comm-density}. App.~\ref{app:classical_ringelmann} relates the form to classical Ringelmann effects.

\emph{Contributions.} \textbf{C1 (regime-classifying scaling law)} is supported by App.~\ref{app:corr_level} (correctness-level fits and $\rho$ comparisons), App.~\ref{app:diversity} (answer-diversity metrics), and the agent-vs-team paradox in App.~\ref{app:transitions}.  \textbf{C2 ($k\tau$ product theorem)} is grounded by App.~\ref{app:k_sweep} ($k\tau$ collapse and low-agreement stratification) and App.~\ref{app:two_rate} (asymmetric conformity rates and transition diagnostics for assumption A2). \textbf{C3 (shared form across conditions)} is supported by App.~\ref{app:hetero} (heterogeneous teams), App.~\ref{app:cross_model} (master fits across three open-weight families plus Gemini, $4\times$ scale, held-out audit), and App.~\ref{app:thinking} (thinking model). \textbf{C4 (initial correlation as the lever)} is detailed in \S\ref{app:noise_detail} (noise-placebo controls and the inert-peer paradox).

\emph{Deployment.} App.~\ref{app:deployment} bundles deployment guidance: small-$N$ measurement procedure (\S\ref{app:measurement}), token-cost scaling, heuristic accuracy projection, and economic argument.

\emph{Reproducibility.} App.~\ref{app:experimental_details} documents prompts, generation parameters, answer extraction, team aggregation, and hardware. \S\ref{app:statistics} covers bootstrap, fitting, and significance methodology. \S\ref{app:scope} states the formal scope and limitations. App.~\ref{app:aggregator} reports aggregator-robustness checks.

\section{Derivation of the effective-team model}
\label{app:effective_team_model}

\noindent\textit{Notation guide for $\rho$ variants used throughout the paper.}
Subscript on $\rho$ indicates a qualifier (mf, eff, ICC, cond) or team size $N$. Bare $\rho$ denotes the equicorrelation-scale pairwise agreement (\S~\ref{sec:framework}). Superscript with parens indicates round: $\rho^{(0)}$ is the initial (pre-communication) correlation and $\rho^{(\tau)}$ is the correlation at round $\tau$.
\begin{itemize}[nosep]
\item $\rho^{(0)}$: pre-communication pairwise answer correlation, baseline before any peer interaction.
\item $\rho^{(\tau)}$: pairwise correlation at round $\tau$, used in the dynamic mean-field model (Thm.~\ref{thm:comm-density}).
\item $\rho_N$: post-debate (final-round) pairwise correlation as a function of team size $N$. The empirical Ringelmann profile $\rho_N = c N^{-\beta}$ describes this.
\item $\rho_\text{mf}$: mean-field correlation from
Thm.~\ref{thm:comm-density}, tracking pairwise agreement through $d = 1 - \rho^{(\tau)}$.
\item $\rho_\text{eff}^{\text{ans}}$: equicorrelation implied by inverting
$N_\text{eff}^{\text{ans}} = 2^{H(\mathbf{a})}$ through Eq.~\ref{eq:neff}. The Ringelmann curve is estimated against this quantity.
\item $\rho_\text{ICC}^{\text{corr}}$: marginal ICC of binary correctness
indicators across items ($(\mathbb{E}[X_iX_j]-\bar{p}^2)/(\bar{p}(1{-}\bar{p}))$).
\item $\rho$: pairwise answer agreement on the equicorrelation
scale, $(N \cdot A_\text{pair} - 1)/(N - 1)$.
\item $\rho_\text{cond}^{\text{corr}}$: correctness ICC conditioned on round-1 item difficulty (App.~\ref{app:corr_level}), isolating within-item social dependence from shared item-difficulty effects.
\end{itemize}

\noindent\textit{Other notation.}
\begin{itemize}[nosep]
\item $\mathbf{a} = (a_1, \ldots, a_N)$: vector of the $N$ agents' final
answers in a single team trial. Each $a_i$ is the post-debate answer of agent $i$ on a single item, extracted from the agent's free-form output (App.~\ref{app:experimental_details}).
\item $H(\mathbf{a})$: Shannon entropy of $\mathbf{a}$ in bits, $H = -\sum_v p_v \log_2 p_v$ with $p_v$ the empirical frequency of answer value $v$ in the team.
\item $N_\text{eff}^{\text{ans}} = 2^{H(\mathbf{a})}$: answer-level effective team size. Ranges from $1$ (full consensus, $H{=}0$) to $N$ (all agents disagree, $H{=}\log_2 N$). Per-team quantity, reported as a mean across items.
\item $N_\text{eff}^{\text{corr}} = N / (1 + (N{-}1) \rho_\text{ICC}^{\text{corr}})$:
correctness-level effective team size, derived from the Kish design effect on the binary correctness indicator (App.~\ref{app:corr_level}).
\item $R(N) = N_\text{eff}/N$: efficiency ratio. The Ringelmann fit
$R(N) = 1/(1 + c(N{-}1)N^{-\beta})$ targets this quantity.
\end{itemize}

\subsection{From equicorrelated votes to effective team size}
\label{app:variance-matching} 
Let $X_i \in \{0,1\}$ indicate whether agent $i$'s final answer is correct after debate. Write
\[
\mathbb{E}[X_i] = p, \qquad \mathrm{Corr}(X_i, X_j) = \rho_\text{eff} \quad (i \neq j),
\]
where $p$ is the per-agent post-debate accuracy (the probability that an arbitrary agent finishes the debate at the correct answer) and $\rho_\text{eff}$ is the equicorrelation between any pair of agents' final correctness indicators. Under exchangeability,
\[
\mathrm{Var}(X_i) = p(1-p), \qquad \mathrm{Cov}(X_i, X_j) = \rho_\text{eff}\,p(1-p).
\]
For the team mean $\bar X = \frac{1}{N}\sum_{i=1}^N X_i$,
\begin{align}
\mathrm{Var}(\bar X)
&= \frac{1}{N^2}\left[\sum_{i=1}^N \mathrm{Var}(X_i) + \sum_{i \neq j}\mathrm{Cov}(X_i, X_j)\right] \\
&= \frac{1}{N^2}\left[Np(1-p) + N(N-1)\rho_\text{eff}p(1-p)\right] \\
&= \frac{p(1-p)}{N}\big(1 + (N-1)\rho_\text{eff}\big).
\label{eq:appendix_varmean}
\end{align}
Relative to independent votes, the variance is inflated by the Kish design effect $D_\text{eff}(N) := 1 + (N-1)\rho_\text{eff}$. Matching this variance to that of an independent sample of size $N_\text{eff}$ gives $N_\text{eff} = N/(1 + (N-1)\rho_\text{eff})$, which is Eq.~\ref{eq:neff}. This identity is exact for exchangeable Bernoulli votes with common pairwise correlation. When heterogeneous dependence is compressed into a single effective parameter $\rho_\text{eff}$, the formula becomes a variance-matching approximation.

\subsection{Correlated-Condorcet approximation}

Let $S_N := \sum_{i=1}^N X_i$ denote the number of correct final votes. Majority voting is correct when $S_N > N/2$ (strict majority, ties broken by fair coin). Under the same exchangeable model, $\mathrm{Var}(S_N) = Np(1-p)(1 + (N-1)\rho_\text{eff}) = N^2 p(1-p)/N_\text{eff}$. Approximating $S_N$ by a Gaussian yields
\begin{align}
\mathbb{P}(S_N > N/2)
&\approx
\Phi\!\left(
\frac{p - 1/2}{\sqrt{p(1-p)/N_\text{eff}}}
\right),
\label{eq:appendix_cc}
\end{align}
where $\Phi$ is the standard normal CDF. Thus the correlated-vote system is approximated by an independent-vote system with only $N_\text{eff}$ effective voters.

In the paper we use $\text{PP}(N_\text{eff}, p)$ as shorthand for this variance-matched correlated-Condorcet approximation. The experiments observe answer distributions, from which we derive answer-level effective-diversity proxies such as $N_\text{eff}^{\text{ans}}=2^{H(\mathbf{a})}$ and pairwise answer agreement. Eq.~\ref{eq:neff} therefore bridges the answer-level diversity we measure with the correctness-level ceiling we predict.

\subsection{Proposition~\ref{prop:effective-scaling} and proof}
\label{sec:proposition2}
\begin{proposition}[Exact Criterion for Effective Scaling]
\label{prop:effective-scaling}
Within the effective-team model with $\rho_N \ge 0$: if $\rho_N \to \rho_\infty > 0$, then $N_\text{eff}(N) \to 1/\rho_\infty$ (finite ceiling). Linear effective scaling $N_\text{eff}(N)=\Theta(N)$ is possible if and only if $\rho_N = O(1/N)$.
\end{proposition}

\noindent\emph{Proof.}
Starting from Eq.~\ref{eq:neff}, $N_\text{eff}(N)/N = 1/(1 + \frac{N-1}{N} \cdot N \cdot \rho_N)$. Since $(N-1)/N \to 1$, the asymptotic behavior is determined by $N \cdot \rho_N$. If $\rho_N \to \rho_\infty > 0$, then $N_\text{eff}(N) \to 1/\rho_\infty$. $N_\text{eff}(N)=\Theta(N)$ holds iff $N \cdot \rho_N=O(1)$, that is, $\rho_N = O(1/N)$.

\subsection{Derivation of Proposition~\ref{prop:ringelmann-regimes}}
\label{app:proposition1}

Assume $\rho_N = cN^{-\beta} + o(N^{-\beta})$ with $c>0$. From Eq.~\ref{eq:neff}, the Kish denominator expands to $1 + (N-1)\rho_N = 1 + cN^{1-\beta} + o(N^{1-\beta})$. The proof tracks which term dominates as $N \to \infty$ in each of the three regimes.

\emph{$\beta = 0$ (hard ceiling).} $\rho_N \to c$, so $(N-1)\rho_N \approx cN$ dominates the $1$ and $N_\text{eff}(N) = N/(1+cN+o(N)) \to 1/c$.

\emph{$0 < \beta < 1$ (sublinear).} The $cN^{1-\beta}$ term still dominates the $1$ (since $1-\beta > 0$), giving $N_\text{eff}(N) \sim N/(cN^{1-\beta}) = N^\beta/c$.

\emph{$\beta \ge 1$ (linear).} For $\beta = 1$, both terms are $O(1)$ and $N_\text{eff}(N) \to N/(1+c)$. For $\beta > 1$, $cN^{1-\beta} \to 0$, the $1$ dominates, and $N_\text{eff}(N) \to N$ (correlation falls fast enough that voters are asymptotically independent).

\begin{corollary}[Heterogeneity Shifts the Ceiling]
\label{cor:hetero}
Consider two debate systems with correlation profiles $\rho_N^{(1)} \sim c_1N^{-\beta}$ and $\rho_N^{(2)} \sim c_2N^{-\beta}$ sharing the same exponent $\beta$. If $c_2 < c_1$, then system~2 has strictly larger asymptotic effective team size. In the hard-ceiling regime $\beta=0$, the ceiling rises from $1/c_1$ to $1/c_2$.
\end{corollary}

\noindent\emph{Proof.}
Immediate from $N_\text{eff} \sim (1/c)N^\beta$ for $\beta < 1$ and $N_\text{eff} \to 1/c$ for $\beta = 0$. \hfill$\square$

\subsection{Proposition~\ref{prop:generic-ceiling}: hard ceiling is generic}

\begin{proposition}[Hard Ceiling Is Generic]
\label{prop:generic-ceiling}
Within the mean-field model (Eq.~\ref{eq:rho}), if the conformity rate $\alpha \in (0,1)$ and initial correlation $\rho^{(0)} \in [0,1)$ do not depend on team size $N$, then post-debate correlation $\rho_\text{mf}$ is constant in $N$, giving $\beta = 0$ (hard-ceiling regime) with ceiling $N_\text{eff}^* = 1/\rho_\text{mf}$.

Let $\alpha_N$ and $\rho^{(0)}_N$ denote $\alpha$ and $\rho^{(0)}$ at team size $N$. Escaping the hard ceiling ($\beta > 0$) requires, for fixed $\tau$, that \emph{both} $\alpha_N \to 0$ and $\rho^{(0)}_N \to 0$ as $N \to \infty$. Neither alone suffices: if $\alpha_N \to 0$ but $\rho^{(0)}_N \to \rho_\infty > 0$, then $\rho_\text{mf}^{(\tau)} \to \rho_\infty > 0$. If $\rho^{(0)}_N \to 0$ but $\alpha$ stays positive, then $\rho_\text{mf}^{(\tau)} \to 1 - (1{-}\alpha)^\tau > 0$.
\end{proposition}

\noindent\emph{Proof.}
When $\alpha$ and $\rho^{(0)}$ are $N$-independent, $\rho_\text{mf}^{(\tau)}$ is constant in $N$, so $\beta=0$. The two failure modes stated in the proposition follow by direct substitution into Eq.~\ref{eq:rho}, and only when both vanish can $\rho_\text{mf}^{(\tau)} \to 0$. \hfill$\square$

\subsection{Proof of Thm.~\ref{thm:comm-density} (mean-field convergence)}
\label{app:thm-proof}
We prove the recurrence under assumptions (A1)--(A3). Let $q^{(\tau)}$ denote the modal share and $\varepsilon^{(\tau)} = 1-q^{(\tau)}$ the non-modal mass.

\emph{Step~1: Recurrence for non-modal mass.} Under (A1), a non-modal agent samples $k$ peers without replacement. Let $M$ denote the number of modal peers observed. For fixed $k$ as $N \to \infty$, the hypergeometric sampling distribution converges to $\text{Bin}(k, q^{(\tau)})$ with deviation $O(k/N)$, and we use the binomial form throughout. Conditional on $M$, the agent retains its non-modal answer with probability $(1-\alpha_1)^M$ (each modal peer independently fails to trigger adoption). By the binomial PGF, $\mathbb{E}[(1{-}\alpha_1)^M] = \bigl(1 - \alpha_1 q^{(\tau)}\bigr)^k$. Under (A2), modal agents do not switch. The non-modal mass therefore evolves as
\begin{equation}
\varepsilon^{(\tau+1)}
= \varepsilon^{(\tau)}\,(1-\alpha_1 + \alpha_1\,\varepsilon^{(\tau)})^k.
\label{eq:eps_exact}
\end{equation}
Under (A3), the mode label is consistent across rounds.

\emph{Step~2: High-agreement expansion.} For $\varepsilon^{(\tau)} \ll 1$, Taylor-expand $(1{-}\alpha_1 + \alpha_1\varepsilon)^k = (1{-}\alpha_1)^k\bigl(1 + \tfrac{k\alpha_1}{1{-}\alpha_1}\varepsilon + O(\varepsilon^2)\bigr)$, so $\varepsilon^{(\tau+1)} = \varepsilon^{(\tau)}(1{-}\alpha_1)^k + O((\varepsilon^{(\tau)})^2)$.

\emph{Step~3: Convert non-modal mass to pairwise disagreement.} For two agents drawn uniformly, $d = 1 - \sum_a p_a^2 = 1 - q^2 - \sum_{a \neq \text{mode}} p_a^2$. The non-modal term satisfies $0 \le \sum_{a\neq\text{mode}} p_a^2 \le \varepsilon^2$, so $d = 2\varepsilon + O(\varepsilon^2)$ and $\varepsilon = d/2 + O(d^2)$. Substituting: $d^{(\tau+1)} = d^{(\tau)}(1{-}\alpha_1)^k + O((d^{(\tau)})^2)$.

\emph{Step~4: Iteration.} Let $r := (1{-}\alpha_1)^k \in (0,1)$. By induction, $d^{(\tau)} = r^\tau d^{(0)}(1 + O(d^{(0)}))$. Therefore $\rho^{(\tau)} \approx 1 - (1{-}\rho^{(0)})(1{-}\alpha_1)^{k\tau}$ for $d^{(0)} = 1{-}\rho^{(0)} \ll 1$, which is Eq.~\ref{eq:rho-ext}. \hfill$\square$

\subsubsection*{Remark on multi-class extensions and modal-switching robustness}
For multi-class answer spaces, (A1)--(A2) coarsen the dynamics to $\{\text{modal}, \text{non-modal}\}$. Non-modal-to-non-modal transitions preserve total non-modal mass $\varepsilon$ but redistribute minority answers. In convergent dynamics this typically concentrates the minority distribution, which decreases pairwise disagreement $d = 1 - \sum_a p_a^2$. The estimated per-peer learning rate $\alpha_l$ on GSM-Hard debate ranges from $0.014$ at $N{=}30$ to $0.062$ at $N{=}5$ (Table~\ref{tab:transitions}, App.~\ref{app:two_rate}). The $\alpha_1$ in this proof is the \emph{effective} rate absorbing both modal-copy and minority-redistribution effects.

Empirically, (A2) is also approximate: modal agents occasionally adopt non-modal peer answers at a small per-peer rate. We denote this rate $\alpha_{s,1}$, paralleling the per-peer learning rate $\alpha_1$. The notation distinguishes it from the agent-level sycophancy rate $\alpha_s$ of App.~\ref{app:two_rate}. Including $\alpha_{s,1} > 0$ in the recurrence of Step~1 yields
\[
\varepsilon^{(\tau+1)} \approx \varepsilon^{(\tau)}\bigl[(1-\alpha_1)^k + k\alpha_{s,1}\bigr] + O\bigl((\varepsilon^{(\tau)})^2\bigr),
\]
which preserves the geometric contraction and the $k\tau$ product structure as long as $\alpha_{s,1} \ll \alpha_1$. Empirically the unconditional per-round modal-switching rate is $\le 0.025$ on every evaluated task, and dividing by the typical peer count gives $\alpha_{s,1} \lesssim 0.003$, roughly an order of magnitude below $\alpha_1$, so the deviation from Eq.~\ref{eq:contraction} is negligible.

\subsubsection*{Scope: when the theorem applies}
The theorem describes contraction under iterated peer-influenced revision. It applies directly to debate (including sparse-$k$ and heterogeneous variants, Proposition~\ref{prop:hetero}). For conditions without peer exposure (self-correction, noise), $\alpha_1 \to 0$ and the theorem predicts $\rho^{(\tau)} \approx \rho^{(0)}$ (no systematic contraction). The empirical match between these placebo conditions (\S~\ref{sec:results}) is consistent with this null prediction. For self-consistency ($\tau = 0$), the theorem gives $\rho^{(0)} = \rho^{(0)}$ directly.

\subsubsection*{Two scales, one prediction}
Consequence~(iii) of the theorem lives on the \emph{pairwise agreement} scale $\rho$, the theorem's native variable. On free-form tasks the entropy-derived scale ($N_\text{eff}^{\text{ans}} = 2^H$) inflates with $N$ because wrong answers spread across many distinct numeric strings, so the entropy-scale exponent $\beta_\text{ent}$ is not the quantity the theorem predicts. Empirically, $\beta_\text{agr} \approx 0.09$--$0.17$ at fixed $k$ on GSM-Hard (close to the theorem's $\beta_\text{agr} \approx 0$ prediction), while $\beta_\text{ent} \approx 0.25$--$0.39$ on the same data. On MCQA tasks the two scales coincide (App.~\ref{app:k_sweep}).

\subsection{Block-exchangeable extension for heterogeneous teams}
\label{app:hetero_theory}

Thm.~\ref{thm:comm-density} assumes exchangeable agents. Heterogeneous teams violate exchangeability: within-family pairs (such as Qwen--Qwen) agree more than cross-family pairs (such as Qwen--Llama). We extend the framework to a block-exchangeable model.

\begin{proposition}[Heterogeneous Effective Correlation]
\label{prop:hetero}
Consider $N$ agents partitioned into $G$ groups of sizes $n_1, \ldots, n_G$ ($\sum_g n_g = N$) with $\mathrm{Corr}(X_i, X_j) = \rho_W$ within groups and $\rho_B$ between groups ($\rho_W \ge \rho_B \ge 0$). Then
\[
N_\textup{eff} = \frac{N}{1 + (\bar n - 1)\rho_W + (N - \bar n)\rho_B}, \qquad \bar n := \frac{1}{N}\sum_{g=1}^G n_g^2.
\]
Equivalently, the team behaves as a homogeneous team with effective correlation $\rho_\textup{eff} = [(\bar n {-} 1)\rho_W + (N {-} \bar n)\rho_B]/(N{-}1)$.
\end{proposition}

\noindent\emph{Proof.}
$\mathrm{Var}(\bar X) = (1/N^2)[\sum_i \mathrm{Var}(X_i) + \sum_{i \ne j} \mathrm{Cov}(X_i,X_j)]$. The covariance sum decomposes: $\sum_g n_g(n_g{-}1)\rho_W\sigma^2 + [N(N{-}1) - \sum_g n_g(n_g{-}1)]\rho_B\sigma^2$. Using $\sum_g n_g(n_g{-}1) = N(\bar n {-} 1)$, variance-matching against $(p(1{-}p)/N_\textup{eff})$ yields the result. \hfill$\square$

\smallskip\noindent\textbf{Consequences.}
\emph{(i)}~When $\rho_B = \rho_W$, the formula collapses to homogeneous $N_\textup{eff} = N/(1{+}(N{-}1)\rho_W)$. \emph{(ii)}~$\rho_B < \rho_W$ implies $\rho_\textup{eff} < \rho_W$, so the ceiling $1/\rho_\textup{eff}$ rises: cross-family independence lifts the effective team size. \emph{(iii)}~If $\rho_W, \rho_B$ are $N$-independent and groups grow proportionally ($\bar n / N$ constant), then $\rho_\textup{eff}$ is $N$-independent, giving $\beta = 0$.

\smallskip\noindent\textbf{Empirical validation.}
On MMLU-Hard heterogeneous teams (Qwen$+$Llama$+$Ministral, equal-sized groups) at $N{=}30$, within-family pairwise agreement gives $\rho_W = 0.879$ and between-family agreement gives $\rho_B = 0.834$. The predicted effective correlation under equal groups is $\rho_\textup{eff}^{\textup{pred}} = \tfrac{1}{3}\rho_W + \tfrac{2}{3}\rho_B = 0.849$, matching the observed overall $\rho_\textup{eff} = 0.843$ within $0.006$.\footnote{$\rho_W$, $\rho_B$, and the predicted $\rho_\textup{eff}$ are on the pairwise-agreement scale, consistent with the variance-matching derivation. Table~\ref{tab:hetero} reports $\rho_\textup{eff}^{\textup{ans}}$ on the entropy-inverted scale (App.~\ref{app:cost}), which is a different metric.} This confirms that heterogeneous teams lower $c$ (from ${\approx}0.85$ homogeneous to ${\approx}0.54$) because cross-family independence reduces the effective correlation, as Proposition~\ref{prop:hetero} predicts.

\section{Relation to classical Ringelmann effects}
\label{app:classical_ringelmann}

We use the Ringelmann label because the same \emph{process-loss structure} appears in our LLM experiments: actual productivity falls below matched potential productivity as group size grows \citep[recovered and translated by][]{kravitz1986ringelmann}. This is a structural analogy, not a claim that LLM agents reproduce human social psychology or social loafing.
\begin{figure}[H]
\centering
\includegraphics[width=\textwidth]{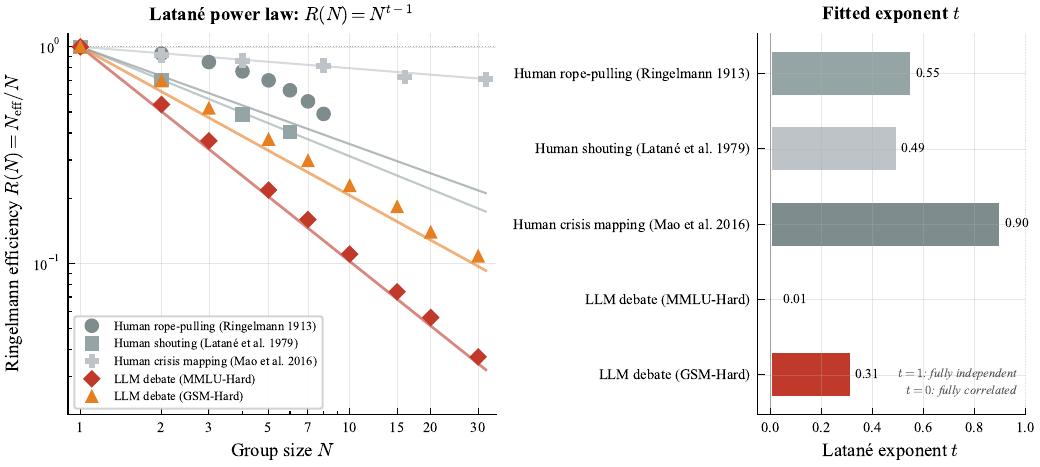}
\caption{\textbf{Structural comparison to classical Ringelmann data
via the Latan\'{e} power law.} We fit $R(N) = N^{t-1}$ to LLM debate and to three published human group datasets, on the same axes. The estimated exponent $t$ quantifies the decline rate: human groups have $t \in [0.49, 0.90]$ (gradual decline), LLM debate has $t \le 0.13$ (much steeper). The same form is used for both populations so the gap in estimated $t$ is structural, not an artifact of model choice. Latan\'{e}'s power law matches the Ringelmann form's asymptote at large $N$ but is not a strict special case (see this appendix's introduction for the relationship).}
\label{fig:latane}
\end{figure}

\begin{figure}[H]
\centering
\includegraphics[width=\textwidth]{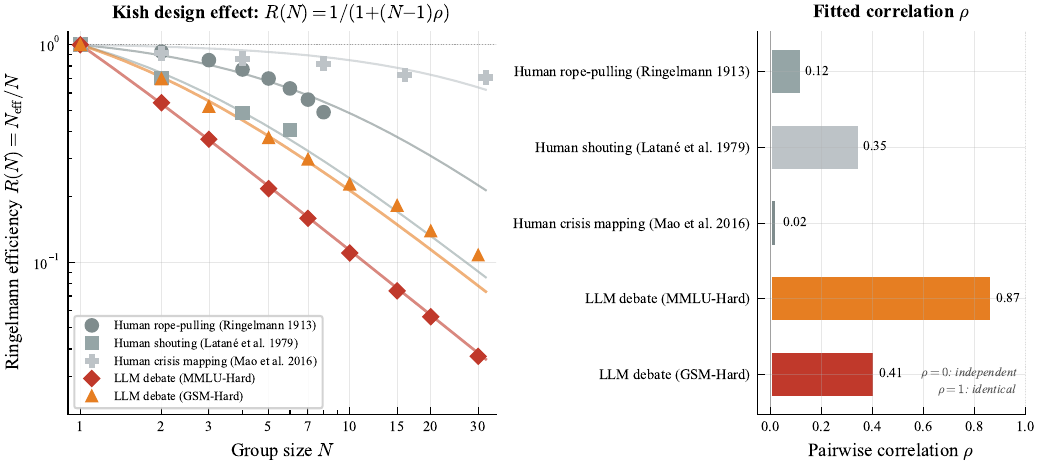}
\caption{\textbf{Structural comparison via the Kish design effect.}
Same datasets as Figure~\ref{fig:latane}, estimated with the constant-$\rho$ form $R(N) = 1/(1{+}(N{-}1)\rho)$, the $\beta = 0$ special case of the Ringelmann form used in the main text. The estimated pairwise correlation $\rho$ quantifies redundancy: human groups have $\rho \in [0.02, 0.35]$ (largely independent voters), LLM debate has $\rho \in [0.56, 0.85]$ (highly correlated). The Kish form is sufficient for the cross-population comparison shown here because both LLM debate and the human studies sit close to the hard-ceiling regime ($\beta \approx 0$). The main analyses use the full two-parameter form because $\beta$ becomes the regime classifier on conditions like self-correction and noise where $\beta > 0$.}
\label{fig:kish}
\end{figure}

Figures~\ref{fig:latane} and~\ref{fig:kish} fit two single-parameter forms used in the human Ringelmann literature: the Latan\'{e} power law $R(N) = N^{t-1}$ and the Kish design effect $R(N) = 1/(1{+}(N{-}1)\rho)$. Kish is the $\beta = 0$ special case of the two-parameter Ringelmann form $R(N) = 1/(1{+}c(N{-}1)N^{-\beta})$ used in the main analyses. Latan\'{e} is not a strict special case but matches the Ringelmann asymptote at large $N$. We use these single-parameter forms here for cross-population comparison on a common axis, but the main paper fits the full two-parameter form because $\beta$ is the regime classifier (hard-ceiling, sublinear, or linear). Forcing $\beta = 0$ collapses the regime distinction and fails on sublinear cells (Const-$\rho$ Kish: $8.3\%$ mean held-out relative error vs.\ $5.5\%$ for the two-parameter Ringelmann form, Table~\ref{tab:held_out_audit}).

For reference, we fit the same declining-efficiency families used in the human literature to our LLM data and to three published human group-productivity measurements. Fitting to Ringelmann's \citeyearpar{ringelmann1913} rope-pulling data ($N{=}1$--$8$) yields $t = 0.55$ ($c = 0.12$). Fitting to the Latan\'{e} et al.\ \citeyearpar{latane1979} shouting data ($N{=}1$--$6$) yields $t = 0.49$ ($c = 0.35$). We also include the \citet{mao2016experimental} crisis-mapping study ($N{=}1$--$32$). For LLM debate, the estimated exponents are much smaller: $t = 0.01$ on MMLU-Hard (implied constant $\rho = 0.85$) and $t = 0.13$ on GSM-Hard (implied $\rho = 0.56$, where this single-parameter Kish fit differs slightly from the two-parameter generalized Ringelmann $c = 0.57$ in the main text). Figure~\ref{fig:latane} plots all five datasets on the Latan\'{e} axis and Figure~\ref{fig:kish} on the Kish axis. Even the mildest human Ringelmann effect retains far more per-person efficiency than LLM debate at matched group sizes.

\section{Experimental details}
\label{app:experimental_details}

The experimental set-up is provided in \S\ref{sec:setup} while in this section additional clarification details regarding the experimentation are available.

\subsection{Model scope}
\label{sec:modelscope}

\textbf{Justifying the model scope of \S\ref{sec:setup}.} Table~\ref{tab:evidence_strength} makes the role of each experiment explicit. The sweep is layered, not a flat cross-product. A single core identification (Qwen2.5-7B, full $N{=}2$--$30$ grid, three conditions, full item budgets of 724 / 1017 / 198) carries the main scaling-law fit; two cross-family replications at matched scope (Llama-3.1-8B and Ministral-8B) test that the form is not Qwen-specific; and four boundary experiments each probe one framework prediction at narrower but sufficient coverage having $4\times$ within-family scale (Qwen2.5-32B), architectural diversity (heterogeneous teams), reasoning paradigm (Qwen3-8B in thinking mode), and closed-source paradigm (Gemini Flash-Lite). The frontier solo evaluation (Gemini 3.1 Pro at $N{=}1$ only) is included as a witness for the predicted-saturated regime rather than as part of the law's identification, since the framework predicts $c \to 1$ at its solo-accuracy level (\S\ref{sec:dynamics}). Item budgets and team-size grids are matched to the question each cell poses: full grids where extrapolation is audited, smaller item subsets where the test is qualitative (regime preservation, ceiling shift), and single-$N$ where the framework already predicts a degenerate fit. The 44 $(c, \beta)$ estimates that result therefore allocate marginal cost in proportion to the strength of the claim each cell supports.

\begin{table}[H]
\caption{Strength of evidence per experiment.
Core sweeps (full $N$ range, full item budget) carry the main claim. Replications, boundary checks, and the frontier solo evaluation play supporting roles. Items column: per-task item counts (M = MMLU-Hard, G = GSM-Hard, Q = GPQA Diamond). Cond.: d = debate, s = self-correction, n = noise placebo.}
\label{tab:evidence_strength}
\centering
\small
\begin{tabular}{llllll}
\toprule
Experiment & Tasks & Items & $N$ & Cond. & Role \\
\midrule
Qwen2.5-7B & MGQ & 724/1017/198 & 2--30 & d/s/n & Open-weight, full sweep \\
Llama-3.1-8B & MGQ & 724/1017/198 & 2--30 & d/s/n & Open-weight, full sweep \\
Ministral-8B & MGQ & 724/1017/198 & 2--30 & d/s/n & Open-weight, full sweep \\
Qwen2.5-32B & MGQ & 200/200/198 & 2--30 & d/s/n & $4\times$ scale \\
Heterogeneous & MGQ & 200/200/200 & 3--30 & d & Diversity \\
Qwen3-8B-think & M & 724 & 1--5 & d/s/n & Thinking \\
Gemini Flash-Lite & MGQ & 200/198/198 & 2--30 & d/s/n & Closed-source \\
Gemini 3.1 Pro & MGQ & 200/211/198 & 1 & ensemble & Frontier solo \\
\bottomrule
\end{tabular}
\end{table}

\subsection{Prompt templates}
\label{sec:prompts}

All agents receive structured prompts that request a final answer, confidence score, and brief rationale. The format is identical across conditions, and only the peer-response block varies.

\textbf{Initial solve (MCQA):}
\begin{verbatim}
Answer the multiple-choice question. Be concise.

Question: {question}
{choices}

Format EXACTLY (answer and confidence FIRST, then reasoning):
FINAL: <A, B, C, or D>
CONF: <0-100>
RATIONALE: <brief reasoning, max 4 lines>
\end{verbatim}
For math tasks, ``\texttt{FINAL: <integer>}'' replaces the letter constraint, and the rest of the template is analogous.

\textbf{Revision (debate / self-correction / noise):}
\begin{verbatim}
You already solved this problem. Now consider what
other team members said. Carefully evaluate their
reasoning -- they may be wrong. Update or keep your answer.

Question: {question}
{choices}

Your previous answer:
{own_answer}

Peer responses:
{peers}

Format EXACTLY (answer and confidence FIRST, then reasoning):
FINAL: <A, B, C, or D>
CONF: <0-100>
RATIONALE: <brief reasoning, note any disagreements>
\end{verbatim}

The three communication conditions differ only in how \texttt{\{peers\}} is populated:
\begin{itemize}
\item \textbf{Debate} (\texttt{comm\_true}): true peer snippets from the same item.
\item \textbf{Self-correction} (\texttt{comm\_self}): the agent's own previous output,
formatted identically to a peer snippet.
\item \textbf{Noise} (\texttt{comm\_random}): randomly shuffled snippets drawn from
other agents' responses to \emph{different} items.
\end{itemize}

Each peer snippet is formatted as ``\texttt{Agent \{i\}: FINAL: \{answer\} CONF: \{conf\} TEXT: \{reasoning\}}'' and truncated to 120 tokens to keep prompts within context limits at large $N$.

\subsection{Generation parameters}

All experiments use the same generation configuration: temperature~$= 0.4$, top-$p = 0.95$, max new tokens~$= 1024$. Seeded random number generators (keyed on item ID, run ID, and round number) ensure deterministic peer selection and tiebreaking across runs.

\subsection{Answer extraction}

Model outputs are parsed with a priority chain. (1)~\texttt{FINAL:} field (highest priority). (2)~\texttt{ANSWER:} field. (3)~\verb|\boxed{}| expressions. (4)~last matching letter (A--D) for MCQA or last integer for math, as fallback. Confidence is extracted from the \texttt{CONF:} field, clamped to $[0, 100]$. For thinking models, \texttt{<think>...</think>} tags are stripped before parsing. Invalid-rate monitoring flags conditions where $>5\%$ of outputs fail to parse.

\subsection{Team aggregation}

The team answer is determined by plurality vote over valid parsed answers. Ties are broken by a deterministic pseudorandom choice seeded on \texttt{hash(item\_id, run\_id)}, ensuring reproducibility. Confidence scores are \emph{not} used in aggregation.

\subsection{Hardware}

Homogeneous experiments run on NVIDIA H100 GPUs. Heterogeneous experiments run on A100 GPUs. Some Llama-3.1-8B sweeps complete on L40S GPUs. The full research project consumed approximately $2{,}000$ GPU-hours across H100, A100, and L40S devices, including preliminary and superseded configurations not reported in the main results. We report cost in prompt and output token counts rather than absolute joules, which makes the comparison hardware-independent and removes batch-utilization artifacts (a small microbatch underutilizing an H100 inflates J/token by a factor unrelated to the scaling claim).

\subsection{Statistical methodology}
\label{app:statistics}

\textbf{Bootstrap confidence intervals.} All reported $95\%$ intervals on estimated $(c, \beta)$ parameters, $N_\text{eff}$ values, and Steiner decomposition components use item-level non-parametric bootstrap with at least $300$ iterations (more where bootstrap precision matters: $1{,}000$ for extrapolation envelopes, $5{,}000$ for paper-number sanity checks). For each iteration we resample items with replacement, recompute the per-cell $N_\text{eff}/N$ profile, refit the Ringelmann curve with regime-constrained bounds ($c \in [0, 5]$, $\beta \in [0, 1]$), and collect the percentile interval. The bootstrap uses a tighter upper bound on $c$ than the point estimation in \S~\ref{sec:setup} ($c \in [0, 20]$) for optimizer stability under resampling. All observed estimates are $c \le 1$, so the tighter bound is non-binding in practice. Item-level resampling preserves the within-item correlation structure that drives $\rho$, which would be destroyed by agent-level resampling.

\textbf{Parameter estimation.} The Ringelmann form $R(N) = 1/(1 + c(N{-}1)N^{-\beta})$ is estimated by non-linear least squares (\texttt{scipy.optimize.curve\_fit}) on the observed $N_\text{eff}/N$ values across the available team sizes, with bounds matching the regime classification of \S~\ref{sec:framework}. The held-out audit (App.~\ref{app:cross_model}) uses the identical bounds.

\textbf{Significance tests.} $p$-values reported in the Steiner decomposition table (App.~\ref{app:corr_level}) come from a paired item-level bootstrap of the AP $-$ PP$_\text{rev}$ contrast: positive bootstrap mass below zero is the two-sided $p$-value. No multiple-comparison correction is applied, and the reported tests are pre-registered contrasts (debate vs.\ matched control on three tasks).

\subsection{Scope and limitations}
\label{app:scope}

\textbf{What was evaluated.} The empirical sweep spans four open-weight model families (Qwen2.5-7B/32B, Llama-3.1-8B, Ministral-8B, Qwen3-8B thinking), a closed-source paradigm point (Gemini Flash-Lite), and a frontier solo evaluation (Gemini 3.1 Pro at $N{=}1$ on all three tasks), three tasks (MMLU-Hard, GSM-Hard, GPQA), team sizes $N \in \{1, 2, 3, 5, 7, 10, 15, 20, 30\}$, three communication conditions (debate, self-correction, noise placebo) plus the self-consistency baseline, homogeneous and heterogeneous compositions, and a communication-density ablation $k \in \{1, 2, 4, 9, 29\}$, $\tau \in \{0,\ldots,5\}$ (App.~\ref{app:k_sweep}).

\textbf{Scope by design.} The empirical sweep is intentionally restricted to homogeneous, role-symmetric, fully-connected interaction so that the conformity signal can be isolated under controlled conditions. Asymmetric roles, external tools, and aggregator pipelines all introduce additional mechanisms (verifier feedback, tool-grounded evidence, asymmetric information flow) that would confound the inter-agent correlation we measure. The following configurations are testable extensions of the framework rather than gaps in its empirical claims:
\begin{itemize}[leftmargin=1.4em, itemsep=0pt, topsep=2pt]
\item \emph{Role-asymmetric debate} (judges, critics, moderators). Production systems often use a single underlying LLM playing multiple roles, where the conformity dynamics may persist because role-conditioning does not decompose the model.
\item \emph{Tool-augmented agents} (web search, code execution, calculators, retrieval).
\item \emph{Mixture-of-Agents \citep{wang2024moa} and verifier-aggregator pipelines} with a designated aggregator.
\item \emph{Multimodal tasks} (vision, audio).
\item \emph{Multi-agent frontier-class debate}. Gemini 3.1 Pro is evaluated solo here, since the framework predicts a saturated ceiling at its solo-accuracy regime.
\item \emph{Very large team sizes} ($N > 30$) and \emph{long-horizon debates} ($\tau > 5$).
\end{itemize}

The framework yields a determinate prediction for each: role asymmetry, external tools, and verifiers should lower $\rho^{(0)}$ and raise the ceiling, frontier-class scale should push $c \to 1$ on already-high-accuracy tasks, and very large $N$ and long $\tau$ should extend the $k\tau$-product collapse. Confirming these predictions requires dedicated experimental control rather than relaxing the current setup, and is left for follow-up.

\textbf{Aggregator and team-balance.} The reported $(c, \beta)$ assume plurality (majority vote) aggregation. Re-aggregating the same agent-level data under confidence-weighted, quartile-filtered, or top-1-confident rules leaves $\beta$ within numerical noise of plurality on MCQA and shifts team accuracy by at most $5$~pp across all $(\text{task}, N)$ cells (mean $\le 1.3$~pp), with the largest shifts concentrated at small $N$ and on top-1-confident (App.~\ref{app:aggregator}). Heterogeneous teams use balanced compositions ($N \in \{3, 6, 9, 15, 30\}$ split equally across three families), which makes block-weighted aggregation algebraically identical to plurality and family-then-vote within $1$~pp. \emph{Imbalanced} hetero designs would discriminate the block-exchangeable extension (App.~\ref{app:hetero_theory}) more sharply and remain a natural follow-up.

\textbf{Range bounded at $N{=}30$, $\tau{\le}6$.} The empirical sweep is bounded at these values. The held-out audit (App.~\ref{app:cross_model}) evaluates within-range interpolation but does not certify the law beyond it. Deployments at $N > 30$ should re-pilot to confirm regime stability.

\textbf{Decoding parameters fixed.} All estimates are conditioned on the decoding configuration in App.~\ref{app:experimental_details} (temperature $0.4$, top-$p$ $0.95$, fixed seed). Lower temperatures or restricted top-$k$/top-$p$ would raise $\rho^{(0)}$ (less initial diversity) and tighten the ceiling, and the opposite holds for higher temperature. A systematic decoding-parameter sweep is left to follow-up.

\textbf{Noise placebo, single form.} Our placebo uses shuffled correct/incorrect peer answers (length-matched to debate). Richer controls (token-length-only padding, fluent-but-off-topic text, adversarial peers) would more cleanly isolate the contribution of peer content from prompt-overhead effects, and remain a natural follow-up.

\textbf{$N_\text{eff}^{\text{ans}} = 2^H$ as an upper bound.} Treating answer entropy as a count of equivalent independent voters is heuristic: on free-form tasks, distinct wrong answers inflate $H$ without reducing $\rho$. The GSM-Hard contrast $\beta_\text{ent} \gg \beta_\text{agr}$ in App.~\ref{app:k_sweep} \emph{is} this bias visible directly. We report both scales rather than collapsing them, and the theorem operates on the agreement scale.

\textbf{Domain.} Tasks are bounded-answer (MCQA) or numeric-answer (arithmetic word problems with a unique correct integer). Open-ended generation, multi-step planning, and creative tasks are not covered.

\textbf{Sampling.} GPQA is reported on the Diamond split ($n{=}198$). MMLU-Hard is filtered to items where the underlying model's solo accuracy is below $0.7$. GSM-Hard uses the Chen et al.\ adversarial subset. Per-cell items range from $198$ to $1017$.

\section{Correctness-level analysis}
\label{app:corr_level}

The main paper presents the answer-level law as primary ($N_\text{eff}^{\text{ans}} = 2^{H(\mathbf{a})}$). This appendix supports C1 by presenting the complementary correctness-level law, estimating the same two-parameter form against $N_\text{eff}^{\text{corr}} = N / (1 + (N{-}1)\rho_\text{ICC}^{\text{corr}})$, where $\rho_\text{ICC}^{\text{corr}}$ is the marginal intraclass correlation of binary correctness indicators across items.

\textbf{Why the regime crossing is answer-level only.} The regime crossing on GSM-Hard between debate and self-correction/noise (\S~\ref{sec:universality}) is specifically an answer-level phenomenon: at the correctness level, all GSM-Hard conditions sit in the hard-ceiling regime (Table~\ref{tab:corr_level_fits}). Within the framework, this is the predicted answer-space signature of conformity. Peer influence raises $\rho$ enough to flatten $\beta$ at the answer level, consolidating wrong-answer mass into fewer distinct strings. Correctness-level ceilings, driven by shared item difficulty plus low solo accuracy, are already saturated and do not move. This consolidation appears on GSM-Hard's much larger answer space, not just on 4-choice MCQA, suggesting partial generalization toward open-ended tasks.

\textbf{Comparison of $\rho$ measures.} Table~\ref{tab:c_vs_rho} compares estimated $c$ against the three $\rho$ quantities (answer-level entropy-derived $\rho_\text{eff}^{\text{ans}}$, correctness-level ICC $\rho_\text{ICC}^{\text{corr}}$, and pairwise agreement $\rho$) at $N{=}30$ on Qwen2.5-7B. On MCQA tasks the three agree closely. On free-form GSM-Hard, $\rho_\text{ICC}^{\text{corr}} \gg \rho_\text{eff}^{\text{ans}}$ because many distinct wrong numeric answers inflate answer entropy while correctness redundancy stays high.

\begin{table}[H]
\caption{Scaling-law parameter $c$ compared with answer-level and
correctness-level correlation at $N{=}30$ (Qwen2.5-7B). $\rho_\text{eff}^{\text{ans}}$: equicorrelation implied by $N_\text{eff}^{\text{ans}} = 2^{H(\mathbf{a})}$. $\rho_\text{ICC}^{\text{corr}}$: marginal correctness ICC over binary indicators across items. $\rho$: pairwise answer agreement.}
\label{tab:c_vs_rho}
\centering
\small
\begin{tabular}{llcccc}
\toprule
Task & Condition & Fitted $c$ & $\rho_\text{eff}^{\text{ans}}$ & $\rho_\text{ICC}^{\text{corr}}$ & $\rho$ \\
\midrule
\multirow{3}{*}{MMLU-Hard} & Debate & .85 & .89 & .92 & .94 \\
                            & Self-corr. & .84 & .82 & .87 & .90 \\
                            & Noise & .65 & .72 & .79 & .84 \\
\midrule
\multirow{3}{*}{GPQA}      & Debate & .81 & .86 & .91 & .93 \\
                            & Self-corr. & .78 & .77 & .83 & .87 \\
                            & Noise & .53 & .55 & .66 & .70 \\
\midrule
\multirow{3}{*}{GSM-Hard}  & Debate & .57 & .54 & .91 & .78 \\
                            & Self-corr. & .34 & .14 & .69 & .42 \\
                            & Noise & .33 & .13 & .71 & .41 \\
\bottomrule
\end{tabular}
\end{table}

\textbf{Correctness-level Ringelmann fits.} Table~\ref{tab:corr_level_fits} shows that at the correctness level, \emph{all} conditions fall in the hard-ceiling regime ($\beta \approx 0$). The ``sublinear'' regime observed on GSM-Hard self-correction and noise (Table~\ref{tab:regime_comparison}) is an answer-space phenomenon: many distinct wrong numeric answers sustain non-trivial entropy even at large $N$, but correctness-level redundancy saturates. The qualitative condition ordering is preserved: $c_\text{debate} > c_\text{self-corr} \ge c_\text{noise}$ on all tasks.

\begin{table}[H]
\caption{Ringelmann estimates at the correctness level vs.\ answer level
(Qwen2.5-7B, $N \in \{2,\ldots,30\}$). All correctness-level estimates are hard-ceiling ($\beta \approx 0$). At the answer level, GSM-Hard self-correction and noise are sublinear ($\beta \approx 0.25$), and debate is hard-ceiling on both levels.}
\label{tab:corr_level_fits}
\centering
\small
\begin{tabular}{llcccc}
\toprule
 & & \multicolumn{2}{c}{Answer level} & \multicolumn{2}{c}{Correctness level} \\
\cmidrule(lr){3-4}\cmidrule(lr){5-6}
Task & Condition & $c$ & $\beta$ & $c$ & $\beta$ \\
\midrule
\multirow{3}{*}{MMLU-Hard} & Debate & .85 & .000 & .89 & .000 \\
                            & Self-corr. & .84 & .011 & .88 & .011 \\
                            & Noise & .65 & .000 & .72 & .000 \\
\midrule
\multirow{3}{*}{GSM-Hard}  & Debate & .57 & .014 & .88 & .000 \\
                            & Self-corr. & .34 & .252 & .72 & .013 \\
                            & Noise & .33 & .272 & .70 & .004 \\
\midrule
\multirow{3}{*}{GPQA}      & Debate & .81 & .000 & .87 & .000 \\
                            & Self-corr. & .78 & .006 & .85 & .008 \\
                            & Noise & .53 & .000 & .62 & .000 \\
\bottomrule
\end{tabular}
\end{table}

\textbf{Item-difficulty decomposition.} The marginal $\rho_\text{ICC}^{\text{corr}}$ conflates two sources of agent correlation: (i)~shared item difficulty (all agents tend to get easy items right) and (ii)~within-item social dependence (sycophancy, shared prompt bias). To isolate (ii), we condition on round-1 accuracy $\hat{p}_i^{(1)}$ as an external item-difficulty estimate (round~1 is pre-debate). The conditioned ICC
\[
\rho_\text{cond}^{\text{corr}} = \frac{ \mathbb{E}_i\!\bigl[\tfrac{1}{N(N{-}1)}\sum_{j\neq k}X_j^{(3)}X_k^{(3)} - (\hat{p}_i^{(1)})^2\bigr]}{ \mathbb{E}_i\!\bigl[\hat{p}_i^{(1)}(1{-}\hat{p}_i^{(1)})\bigr]}
\]
strips out the item-difficulty component. Table~\ref{tab:icc_decomp} shows that on MCQA tasks, debate induces substantial within-item social dependence ($\rho_\text{cond}^{\text{corr}} = 0.49$--$0.74$), self-correction induces moderate dependence ($0.33$ on MMLU-Hard), and noise adds near-zero social correlation ($0.04$). The condition ordering \emph{debate $>$ self-correction $>$ noise} is preserved after conditioning. On GPQA, negative $\rho_\text{cond}^{\text{corr}}$ for self-correction ($-0.15$) and noise ($-0.23$) indicates that these conditions inject \emph{more} answer diversity than item difficulty alone would predict, consistent with stochastic re-evaluation redistributing mass across answers. We omit GSM-Hard from the conditioned analysis because the round-1 item-difficulty estimate is too noisy for free-form numeric answers, producing $\rho_\text{cond}^{\text{corr}} > 1$ (not a valid ICC).

\begin{table}[H]
\caption{Marginal vs.\ item-difficulty-conditioned ICC at $N{=}30$.
$\rho_\text{r1}$: round-1 marginal ICC (pre-debate baseline). $\rho_\text{r3}$: round-3 marginal ICC. $\rho_\text{cond}^{\text{corr}}$: round-3 ICC conditioned on round-1 item difficulty.}
\label{tab:icc_decomp}
\centering
\small
\begin{tabular}{llcccc}
\toprule
Task & Condition & $\rho_\text{r1}$ & $\rho_\text{r3}$ & $\rho_\text{cond}^{\text{corr}}$ & Ordering \\
\midrule
\multirow{3}{*}{MMLU-Hard} & Debate & .957 & .916 & .741 & \multirow{3}{*}{D$>$S$>$N \checkmark} \\
                            & Self-corr. & .957 & .865 & .332 & \\
                            & Noise & .957 & .794 & .039 & \\
\midrule
\multirow{3}{*}{GPQA}      & Debate & .923 & .912 & .495 & \multirow{3}{*}{D$>$S$>$N \checkmark} \\
                            & Self-corr. & .923 & .833 & $-.147$ & \\
                            & Noise & .923 & .662 & $-.232$ & \\
\bottomrule
\multicolumn{6}{l}{\footnotesize GSM-Hard omitted: round-1 item-difficulty estimates are too noisy for}\\
\multicolumn{6}{l}{\footnotesize free-form numeric answers, producing $\rho_\text{cond}^{\text{corr}} > 1$ (not a valid ICC).}
\end{tabular}
\end{table}

\textbf{Accuracy implications: Steiner decomposition.} Beyond correctness-level correlation, \S~\ref{sec:correlation} reports that team accuracy is invariant to communication mode (debate vs.\ self-correction). Table~\ref{tab:decomposition} provides the matched-control bookkeeping behind that claim: actual productivity ($\text{AP}$, debate accuracy) versus potential productivity ($\text{PP}_{\text{rev}}$, self-correction as the no-peer matched control), with social process loss $\text{PL}_{\text{social}} = \text{PP}_{\text{rev}} - \text{AP}$. Across all three tasks the gap is within $\pm 2$ pp at $N{=}30$ and the bootstrap intervals straddle zero, consistent with the equivalence test in \S\ref{sec:correlation}. The re-evaluation gain column reports the absolute accuracy lift over the round-1 (pre-communication) baseline.

\begin{table}[H]
\caption{Steiner decomposition at $N{=}30$ (Qwen2.5-7B).
$\text{PP}_{\text{rev}}$: self-correction (matched control). $\text{AP}$: debate. Brackets: 95\% bootstrap CIs.}
\label{tab:decomposition}
\centering
\small
\begin{tabular}{lccccc}
\toprule
Task & $\text{PP}_{\text{rev}}$ & $\text{AP}$ & $\text{PL}_{\text{social}}$ & Re-eval.\ gain & $p$ \\
\midrule
MMLU-Hard & 58.7 [55.1, 62.3] & 57.9 [54.3, 61.5] & 0.8 [$-$1.0, 2.6] & +3.2 pp & 0.478 \\
GSM-Hard  & 28.2 [25.5, 31.0] & 26.4 [23.6, 29.0] & 1.9 [$-$0.2, 3.9] & +18.2 pp & 0.090 \\
GPQA      & 31.3 [25.3, 37.9] & 31.8 [25.8, 38.4] & $-$0.5 [$-$4.5, 3.5] & +1.5 pp & 1.000 \\
\bottomrule
\end{tabular}
\end{table}

\section{Diversity metrics}
\label{app:diversity}

This appendix supports C1 by reporting the answer-diversity scale quantities ($N_\text{eff}^{\text{ans}}$, $\rho_\text{eff}^{\text{ans}}$, $\rho$) on which the Ringelmann fit operates. Table~\ref{tab:diversity} shows their per-condition values at $N{=}30$ on Qwen2.5-7B.

\begin{table}[H]
\caption{Diversity metrics at $N{=}30$, final round (Qwen2.5-7B).
$\rho_\text{eff}^{\text{ans}}$ from entropy-based $N_\text{eff}^{\text{ans}}$, and $\rho$ from pairwise agreement.}
\label{tab:diversity}
\centering
\small
\begin{tabular}{llccccc}
\toprule
Task & Condition & $N_\text{eff}^{\text{ans}}$ & $\rho_\text{eff}^{\text{ans}}$ & $\rho$ & \% Unan. & \% Unan.\ wrong \\
\midrule
\multirow{3}{*}{MMLU} & Debate & 1.12 & .886 & .939 & 80\% & 30\% \\
                       & Self-corr. & 1.21 & .822 & .904 & 69\% & 24\% \\
                       & Noise & 1.37 & .722 & .842 & 48\% & 14\% \\
\midrule
\multirow{3}{*}{GSM} & Debate & 1.79 & .542 & .784 & 47\% & 26\% \\
                      & Self-corr. & 6.10 & .135 & .416 & 12\% & 3\% \\
                      & Noise & 6.44 & .126 & .408 & 11\% & 2\% \\
\midrule
\multirow{3}{*}{GPQA} & Debate & 1.15 & .864 & .928 & 76\% & 48\% \\
                       & Self-corr. & 1.29 & .767 & .865 & 57\% & 35\% \\
                       & Noise & 1.76 & .552 & .703 & 24\% & 11\% \\
\bottomrule
\end{tabular}
\end{table}

\textbf{Key patterns.} The condition ordering $\text{debate} > \text{self-correction} > \text{noise}$ on $\rho$ holds across all three tasks, with the noise placebo close to self-correction throughout (on GSM, $\rho = 0.408$ vs.\ $0.416$). The GSM regime crossing is visible directly in $N_\text{eff}^{\text{ans}}$: debate saturates at $1.79$ while self-correction and noise reach $6.10$ and $6.44$ respectively, whereas on MCQA every condition stays close to the hard ceiling ($N_\text{eff}^{\text{ans}} \le 1.76$). The \%~unanimous-wrong column captures the failure mode no aggregator can recover: it peaks at $48\%$ on GPQA debate and falls to $11$--$14\%$ under noise, consistent with peer influence concentrating mass on correlated errors.

\section{Agent-level transition analysis}
\label{app:transitions}

This appendix supports C1 by resolving an apparent paradox. The team-level scaling collapse documented in \S~\ref{sec:results} is not driven by a failure of agents to learn from peers. The opposite holds: agents \emph{do} update toward correct peer answers, yet the aggregate efficiency $R(N) = N_\text{eff}/N$ collapses steeply with $N$. Figure~\ref{fig:transitions} resolves this paradox.

\textbf{Individual learning is positive.} Counting agent-level wrong-to-correct (W$\to$C) and correct-to-wrong (C$\to$W) transitions between rounds R1 and R3, the net rate W$\to$C $-$ C$\to$W is positive at every team size on every task. On GSM-Hard at $N{=}10$, individual W$\to$C transitions outnumber C$\to$W by roughly $18\times$ (Table~\ref{tab:two_rate} reports the unconditional rates: $0.088$ vs.\ $0.005$). So the conformity dynamics do net \emph{toward} correctness on average. Debate is not a noise process at the agent level.

\textbf{Collective gains do not scale.} The bottleneck is informational, not behavioral. Even though most individual updates are W$\to$C, the redirected agents land on whatever the modal answer is, and on items where the round-1 modal is wrong the modal-leak rate $\alpha_s$ (App.~\ref{app:two_rate}) is non-trivial. Per-item contributions therefore correlate sharply across agents: the team contracts to a single answer rather than aggregating $N$ near-independent votes. The effective team size $N_\text{eff}$ tracks this contraction directly and stays at or below the matched no-peer control across $N{=}2$--$30$.

\textbf{Mechanistic reading.} Two facts stand in tension. Individual learning is positive, while team efficiency $R(N) = N_\text{eff}/N$ collapses with $N$. This is exactly what the Ringelmann scaling law predicts when conformity raises pairwise correlation faster than per-agent accuracy improves. The framework treats these as the same phenomenon viewed at two scales: agent-level W$\to$C transitions (driven by $\alpha_l$) and team-level redundancy (driven by $\rho$). Both are positive, but only the second matters for aggregation, and the second saturates first.

\begin{figure}[H]
\centering
\includegraphics[width=\textwidth]{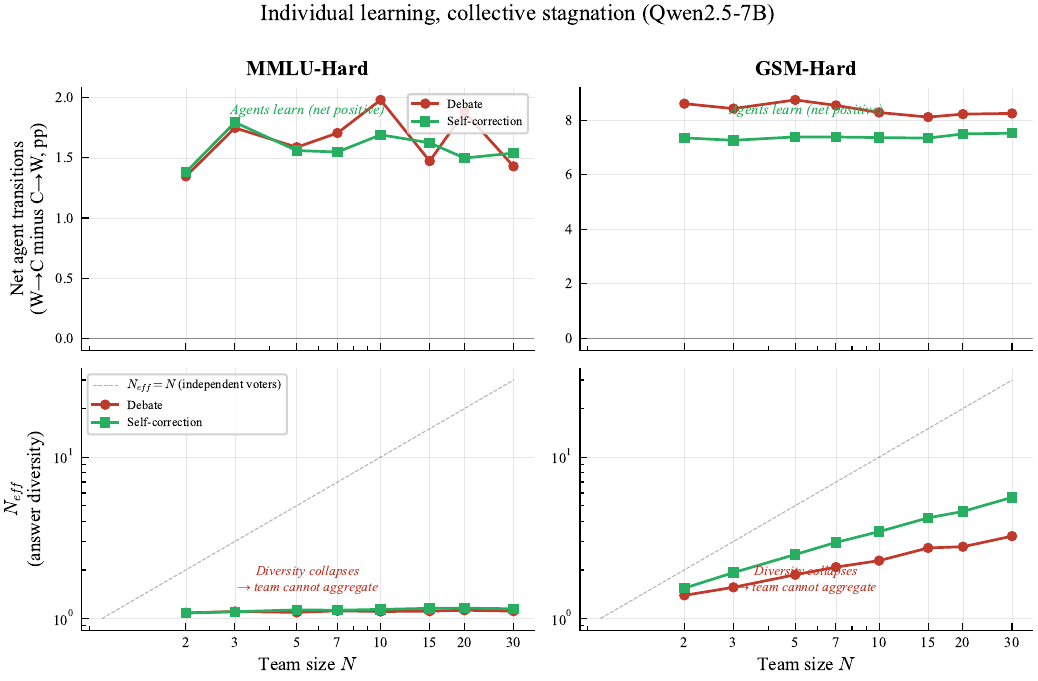}
\caption{\textbf{The Ringelmann paradox: individual learning, collective stagnation.}
\textbf{Top:} Net agent-level transitions (W$\to$C minus C$\to$W) are positive at all team sizes, on both tasks, under both debate and self-correction. Agents do learn from interaction. \textbf{Bottom:} Effective team size $N_\text{eff}$ stays close to the no-aggregation lower bound on MMLU-Hard and grows sublinearly on GSM-Hard, in both cases far below the independent-voter line $N_\text{eff} = N$. Individual learning cannot be aggregated because diversity collapses.}
\label{fig:transitions}
\end{figure}

\section{Heterogeneous teams}
\label{app:hetero}

This appendix supports C3 by evaluating whether the same Ringelmann form applies when teams mix three model families (Qwen, Llama, Ministral) at team sizes $N \in \{3, 6, 9, 15, 30\}$ on MMLU-Hard, GSM-Hard, and GPQA (200 items each). The framework predicts that architectural diversity lowers initial correlation $\rho^{(0)}$ and raises the ceiling $1/c$ while preserving the regime $\beta$. Within-family vs.\ between-family pairwise agreement (Table~\ref{tab:within_between_rho}) also tests the block-exchangeable extension of the variance-matching identity (App.~\ref{app:hetero_theory}).

\begin{table}[H]
\caption{Heterogeneity shifts the Ringelmann ceiling.
Matched-subset comparison at $N{=}30$ (200 items per task).}
\label{tab:hetero}
\centering
\small
\begin{tabular}{llccccc}
\toprule
Task & Setting & $N_\text{eff}^{\text{ans}}$ & $\rho_\text{eff}^{\text{ans}}$ & Fitted $c$ & $\beta$ & Ceiling $1/c$ \\
\midrule
\multirow{2}{*}{MMLU-Hard} & Homogeneous & 1.12 & .886 & .85 & .000 & 1.18 \\
                            & Heterogeneous & 1.34 & .735 & .54 & .000 & 1.86 \\
\midrule
\multirow{2}{*}{GSM-Hard}  & Homogeneous & 1.79 & .542 & .57 & .014 & 1.74 \\
                            & Heterogeneous & 2.25 & .426 & .35 & .000 & 2.85 \\
\midrule
\multirow{2}{*}{GPQA}      & Homogeneous & 1.15 & .864 & .81 & .000 & 1.24 \\
                            & Heterogeneous & 1.27 & .778 & .56 & .000 & 1.80 \\
\bottomrule
\end{tabular}
\end{table}

The matched-subset comparison in Table~\ref{tab:hetero} shows the same declining-efficiency structure as homogeneous debate, with $c$ falling and the asymptotic ceiling $1/c$ rising on every task while $\beta$ stays at $0$. The shift is a ceiling change, not a regime change. Figure~\ref{fig:hetero_extrapolation} shows that the small-$N$ Ringelmann fit also extrapolates to the largest $N$ on heterogeneous data.

\begin{figure}[H]
\centering
\includegraphics[width=\textwidth]{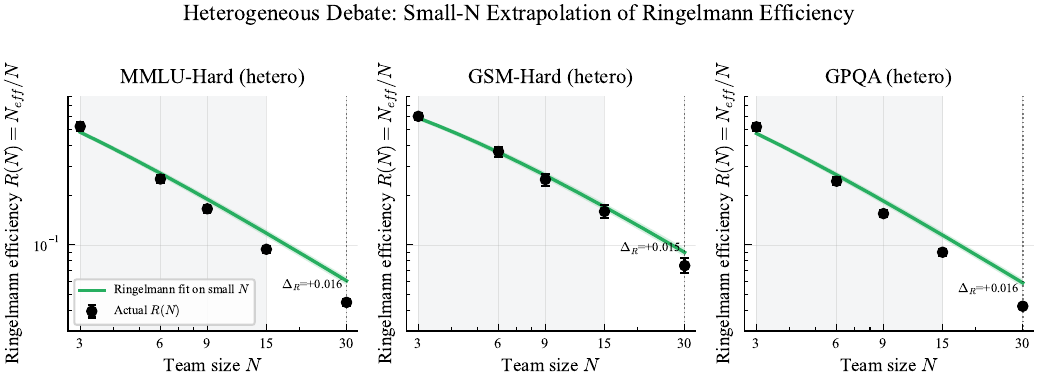}
\caption{\textbf{Small-$N$ extrapolation of the Ringelmann scaling law
on heterogeneous debate teams} across MMLU-Hard, GSM-Hard, and GPQA. Green curves are the Ringelmann form $R(N)=1/(1{+}c(N{-}1)N^{-\beta})$ estimated on small $N$ (training range varies by task), with $95\%$ bootstrap envelopes, extrapolated to the largest $N$ available. The same form continues to fit, with a higher ceiling ($1/c$) than in the homogeneous setting.}
\label{fig:hetero_extrapolation}
\end{figure}

\textbf{Within- vs.\ between-family pairwise agreement.} The block-exchangeable extension (App.~\ref{app:hetero_theory}) predicts that pairs from the same model family should agree more strongly than pairs from different families. We confirm this directly on the heterogeneous data by enumerating all final-round agent pairs within each item and partitioning by family pairing (Table~\ref{tab:within_between_rho}). Aggregated across all cells with $N \ge 6$ (so within-family pairs exist), $\bar\rho_{\text{within}} = 0.824$ vs.\ $\bar\rho_{\text{between}} = 0.761$ ($\Delta = +0.063$), with all six family-pair combinations ranking in the predicted direction: same-family pairs (Qwen-Qwen $0.80$, Llama-Llama $0.77$, Ministral-Ministral $0.80$) agree more than cross-family pairs (Qwen-Llama $0.68$, Qwen-Ministral $0.72$, Llama-Ministral $0.68$). The gap is largest on free-form GSM-Hard ($\Delta$ up to $+0.15$ at $N{=}6$), where families differ most in solution style, and smallest on MCQA where the bounded answer space forces structural agreement. The same hetero teams whose \emph{aggregate} $\rho$ is lower than homogeneous teams (Table~\ref{tab:hetero}) thus exhibit the predicted internal block structure, supporting the variance-matching bridge under the block-exchangeable assumption.

\begin{table}[H]
\caption{Within- and between-family pairwise agreement on
heterogeneous teams (final round, weighted means across all cells with $N{\ge}6$). Same-family pairs agree more than cross-family pairs in every combination, validating the block-exchangeable prediction empirically.}
\label{tab:within_between_rho}
\centering
\small
\begin{tabular}{lcccc}
\toprule
Task & $\bar\rho_{\text{within}}$ & $\bar\rho_{\text{between}}$ &
$\Delta$ & Pairs (W / B) \\
\midrule
MMLU-Hard & 0.868 & 0.819 & $+0.050$ & 35{,}400 / 82{,}800 \\
GSM-Hard  & 0.722 & 0.616 & $+0.105$ & 35{,}400 / 82{,}800 \\
GPQA      & 0.882 & 0.849 & $+0.033$ & 35{,}037 / 81{,}952 \\
\midrule
Overall   & 0.824 & 0.761 & $+0.063$ & 105{,}837 / 247{,}552 \\
\bottomrule
\end{tabular}
\end{table}

\section{Cross-model and cross-scale verification experiments}
\label{app:cross_model}

This appendix supports C3 by extending the core scaling-law identification (\S~\ref{sec:results}, Qwen2.5-7B, full $N{=}2$--$30$, three conditions) across model families, scale, paradigm, and team composition. Each verification experiment below evaluates a specific framework prediction. Table~\ref{tab:evidence_strength} in App. \ref{sec:modelscope}tabulates scope per experiment. Table~\ref{tab:master_fits} summarizes every $(c, \beta)$ estimate. Table~\ref{tab:cross_model} organizes the same estimates by framework prediction.

\begin{table}[H]
\caption{Master table of all estimated Ringelmann parameters $(c, \beta)$.
$R(N) = 1/(1{+}c(N{-}1)N^{-\beta})$. Regime: H = hard-ceiling ($\beta \approx 0$), S = sublinear ($\beta > 0.1$). $\widehat{N}_\text{eff}$: observed effective team size at the largest $N$ evaluated.}
\label{tab:master_fits}
\centering
\small
\begin{tabular}{lllcccc}
\toprule
Model & Task & Cond. & $c$ & $\beta$ & Reg. & $\widehat{N}_\text{eff}$ \\
\midrule
\multicolumn{7}{l}{\emph{Qwen2.5-7B (core, $N{=}2$--$30$)}} \\
Qwen2.5-7B & MMLU & debate & .85 & .00 & H & 1.12 \\
Qwen2.5-7B & MMLU & self-c & .84 & .01 & H & 1.21 \\
Qwen2.5-7B & MMLU & noise  & .65 & .00 & H & 1.37 \\
Qwen2.5-7B & GSM  & debate & .57 & .01 & H & 1.79 \\
Qwen2.5-7B & GSM  & self-c & .34 & .25 & S & 6.10 \\
Qwen2.5-7B & GSM  & noise  & .33 & .27 & S & 6.44 \\
Qwen2.5-7B & GPQA & debate & .81 & .00 & H & 1.15 \\
Qwen2.5-7B & GPQA & self-c & .78 & .01 & H & 1.29 \\
Qwen2.5-7B & GPQA & noise  & .53 & .00 & H & 1.76 \\
\midrule
\multicolumn{7}{l}{\emph{Llama-3.1-8B (cross-family, $N{=}2$--$30$)}} \\
Llama-3.1-8B & MMLU & debate & .62 & .00 & H & 1.21 \\
Llama-3.1-8B & MMLU & self-c & .49 & .00 & H & 1.94 \\
Llama-3.1-8B & MMLU & noise  & .51 & .00 & H & 1.81 \\
Llama-3.1-8B & GSM  & debate & .47 & .00 & H & 1.94 \\
Llama-3.1-8B & GSM  & self-c & .22 & .36 & S & 10.8 \\
Llama-3.1-8B & GSM  & noise  & .18 & .30 & S & 10.4 \\
Llama-3.1-8B & GPQA & debate & .55 & .00 & H & 1.21 \\
Llama-3.1-8B & GPQA & self-c & .44 & .00 & H & 2.12 \\
Llama-3.1-8B & GPQA & noise  & .45 & .00 & H & 2.13 \\
\midrule
\multicolumn{7}{l}{\emph{Ministral-8B (cross-family, $N{=}2$--$30$)}} \\
Ministral-8B & MMLU & debate & .71 & .00 & H & 1.39 \\
Ministral-8B & MMLU & self-c & .58 & .01 & H & 1.74 \\
Ministral-8B & MMLU & noise  & .54 & .00 & H & 1.81$^{p}$ \\
Ministral-8B & GSM  & debate & .38 & .00 & H & 2.45 \\
Ministral-8B & GSM  & self-c & .29 & .24 & S & 5.74$^{q}$ \\
Ministral-8B & GSM  & noise  & .32 & .17 & S & 4.54$^{q}$ \\
Ministral-8B & GPQA & debate & .57 & .00 & H & 1.56 \\
Ministral-8B & GPQA & self-c & .41 & .00 & H & 2.14 \\
Ministral-8B & GPQA & noise  & .37 & .00 & H & 2.36 \\
\midrule
\multicolumn{7}{l}{\emph{Qwen2.5-32B ($4{\times}$ scale, $N{=}2$--$30$)}} \\
Qwen2.5-32B & MMLU & debate & .96 & .00 & H & 1.04 \\
Qwen2.5-32B & MMLU & self-c & .88 & .00 & H & 1.13 \\
Qwen2.5-32B & MMLU & noise  & .88 & .02 & H & 1.21 \\
Qwen2.5-32B & GSM & debate & .77 & .00 & H & 1.24 \\
Qwen2.5-32B & GSM & self-c & .45 & .18 & S & 3.81 \\
Qwen2.5-32B & GSM & noise  & .45 & .18 & S & 3.82 \\
Qwen2.5-32B & GPQA & debate & .97 & .00 & H & 1.03 \\
Qwen2.5-32B & GPQA & self-c & .92 & .04 & H & 1.16 \\
Qwen2.5-32B & GPQA & noise  & .85 & .02 & H & 1.23$^{p}$ \\
\midrule
\multicolumn{7}{l}{\emph{Heterogeneous teams (Qwen+Llama+Ministral, $N{=}3$--$30$)}} \\
Hetero & MMLU & debate & .54 & .00 & H & 1.34 \\
Hetero & GSM  & debate & .35 & .00 & H & 2.25 \\
Hetero & GPQA & debate & .56 & .00 & H & 1.27 \\
\midrule
\multicolumn{7}{l}{\emph{Thinking models}} \\
Qwen3-8B-think & MMLU & debate & 1.00 & .00 & H & 1.02$^c$ \\
Qwen3-8B-think & MMLU & self-c & .89 & .01 & H & 1.11$^c$ \\
Qwen3-8B-think & MMLU & noise  & .88 & .00 & H & 1.10$^c$ \\
Gemini-FL (min.\ think) & GSM  & debate & .95 & .00 & H & 1.04 \\
Gemini-FL (min.\ think) & GPQA & debate & .86 & .00 & H & 1.10 \\
\bottomrule
\multicolumn{7}{l}{\footnotesize $^c$\,$N \le 5$. $^p$\,partial item coverage at largest $N$. $^q$\,largest $N$ available is $N{=}20$.}
\end{tabular}
\end{table}

\textbf{Frontier solo evaluation (Gemini 3.1 Pro).} The master fits table excludes Gemini 3.1 Pro because no Ringelmann curve can be estimated from a single team size. Solo accuracy at $N{=}1$: $97.0\%$ on MMLU-Hard ($n{=}200$), $80.6\%$ on GSM-Hard ($n{=}211$), and $88.9\%$ on GPQA Diamond ($n{=}198$). The framework predicts $c \to 1$ at this accuracy regime (\S\ref{sec:dynamics}), making the implied $1/c$ ceiling negligibly above one effective agent. We therefore omit multi-agent runs on Gemini 3.1 Pro: the framework's prediction makes them uninformative for the law's identification. The solo numbers confirm that the high-solo-accuracy regime is reached at the frontier, complementing Qwen3-8B thinking mode's $c = 1.00$ at $87.2\%$ solo (Table~\ref{tab:master_fits}).

\textbf{The $(c, \beta)$ landscape.} Figure~\ref{fig:universality} plots all $44$ estimates from Table~\ref{tab:master_fits} on the $(c, \beta)$ plane, color-coded by communication condition and shape-coded by model variant. The same form fits every cell, with parameter shifts that classify each into the regime predicted by \S~\ref{sec:framework}.

\begin{table}[H]
\caption{Fitted Ringelmann parameters (debate condition) organized by
framework prediction. $R(N)=1/(1{+}c(N{-}1)N^{-\beta})$ estimated against observed $N_\text{eff}/N$. $\widehat{N}_\text{eff}$: observed (not predicted) effective team size at the largest $N$ evaluated.}
\label{tab:cross_model}
\centering
\small
\begin{tabular}{llccccc}
\toprule
Prediction evaluated & Model $\times$ Task & $c$ & $\beta$ & $R^2$ & $\widehat{N}_\text{eff}$ \\
\midrule
\multicolumn{6}{l}{\emph{Hard ceiling on MCQA (3 families, $N{=}2$--$30$)}} \\
 & Qwen2.5-7B $\times$ MMLU & .85 & .000 & .9995 & 1.12 \\
 & Llama-3.1-8B $\times$ MMLU & .62 & .000 & .9942$^c$ & 1.21 \\
 & Ministral-8B $\times$ MMLU & .71 & .000 & .9996 & 1.39 \\
 & Qwen2.5-7B $\times$ GPQA & .81 & .000 & .9993 & 1.15 \\
 & Llama-3.1-8B $\times$ GPQA & .55 & .000 & .9920 & 1.21 \\
 & Ministral-8B $\times$ GPQA & .57 & .000 & .9904 & 1.56 \\
\midrule
\multicolumn{6}{l}{\emph{Regime structure on free-form math (answer level, $N{=}2$--$30$)}} \\
 & Qwen2.5-7B $\times$ GSM & .57 & .014 & .9989 & 1.79 \\
 & Llama-3.1-8B $\times$ GSM & .47 & .000 & .9993 & 1.94 \\
 & Ministral-8B $\times$ GSM & .38 & .000 & .9999 & 2.45 \\
\midrule
\multicolumn{6}{l}{\emph{Scale / paradigm tightens ceiling on debate}} \\
 & Qwen2.5-32B $\times$ GSM ($4\times$ scale) & .77 & .000 & .9982 & 1.24 \\
 & Gemini-FL $\times$ GSM (frontier API) & .95 & .000 & .9992 & 1.04 \\
\midrule
\multicolumn{6}{l}{\emph{Thinking models: tightest ceiling}} \\
 & Qwen3-8B-think $\times$ MMLU & 1.00 & .000 & --- & 1.02$^a$ \\
 & Gemini-FL (min.\ think) $\times$ GPQA & .86 & .000 & --- & 1.10 \\
\bottomrule
\multicolumn{6}{l}{\footnotesize $^a$\,$N \le 5$.} \\
\multicolumn{6}{l}{\footnotesize $^c$Lowest $R^2$: Llama $N_\text{eff}$ is mildly non-monotonic
  at large $N$, fit overshoots.}
\end{tabular}
\end{table}

\begin{table}[H]
\caption{\textbf{Ringelmann form wins on held-out extrapolation
across three open-weight families and heterogeneous teams.} Each cell is estimated on $N \in \{2,3,5\}$ and predicts $\widehat{N}_\text{eff}$ at held-out $N \in \{7,10,15,20,30\}$ (9 cells per family on Qwen and Llama, 8 on Ministral; 3 cells for heterogeneous, debate only). Across all three open-weight families the Ringelmann form attains $5$--$12\%$ mean relative error while logistic saturation never beats $74\%$, confirming the form is structurally chosen, not curve-fit. Const-$\rho$ Kish is the special case $\beta = 0$ enforced. On heterogeneous teams the regime-constrained fit pins to $\beta = 0$ everywhere, making Ringelmann and Const-$\rho$ Kish numerically identical.}
\label{tab:held_out_audit}
\centering
\small
\begin{tabular}{lcccc}
\toprule
Functional family & Form & Params & Mean RMSE & Mean rel.\ err.\ \\
\midrule
\multicolumn{5}{l}{\emph{Qwen2.5-7B (9 cells)}} \\
\textbf{Ringelmann} & $1/(1{+}c(N{-}1)N^{-\beta})$ & 2 & \textbf{0.13} & \textbf{5.5\%} \\
Const-$\rho$ Kish     & $1/(1{+}(N{-}1)\rho)$        & 1 & 0.37 & 8.3\% \\
Power law            & $a + bN^{c}$                  & 3 & 0.86 & 32.2\% \\
Logistic saturation  & $L/(1{+}e^{k(N{-}N_0)})$      & 3 & 1.69 & 73.4\% \\
\midrule
\multicolumn{5}{l}{\emph{Llama-3.1-8B (9 cells)}} \\
\textbf{Ringelmann} & --- & 2 & \textbf{0.45} & \textbf{11.9\%} \\
Const-$\rho$ Kish     & --- & 1 & 0.66 & 14.1\% \\
Power law            & --- & 3 & 1.02 & 40.8\% \\
Logistic saturation  & --- & 3 & 2.31 & 69.7\% \\
\midrule
\multicolumn{5}{l}{\emph{Ministral-8B (8 cells)$^*$}} \\
\textbf{Ringelmann} & --- & 2 & \textbf{0.18} & \textbf{5.8\%} \\
Const-$\rho$ Kish     & --- & 1 & 0.23 & 6.6\% \\
Power law            & --- & 3 & 0.79 & 31.6\% \\
Logistic saturation  & --- & 3 & 1.64 & 68.8\% \\
\midrule
\multicolumn{5}{l}{\emph{Heterogeneous teams (3 cells, debate only)}} \\
Ringelmann ($=$Kish) & --- & 2 & 0.43 & 27.4\% \\
Const-$\rho$ Kish     & --- & 1 & 0.43 & 27.4\% \\
Power law            & --- & 3 & 0.96 & 38.5\% \\
Logistic saturation  & --- & 3 & 1.34 & 80.6\% \\
\bottomrule
\multicolumn{5}{l}{\footnotesize $^*$Ministral GSM noise excluded: insufficient $N \in \{2,3\}$ coverage to fit the small-team training set.}
\end{tabular}
\end{table}

\begin{figure}[H]
\centering
\includegraphics[width=\textwidth]{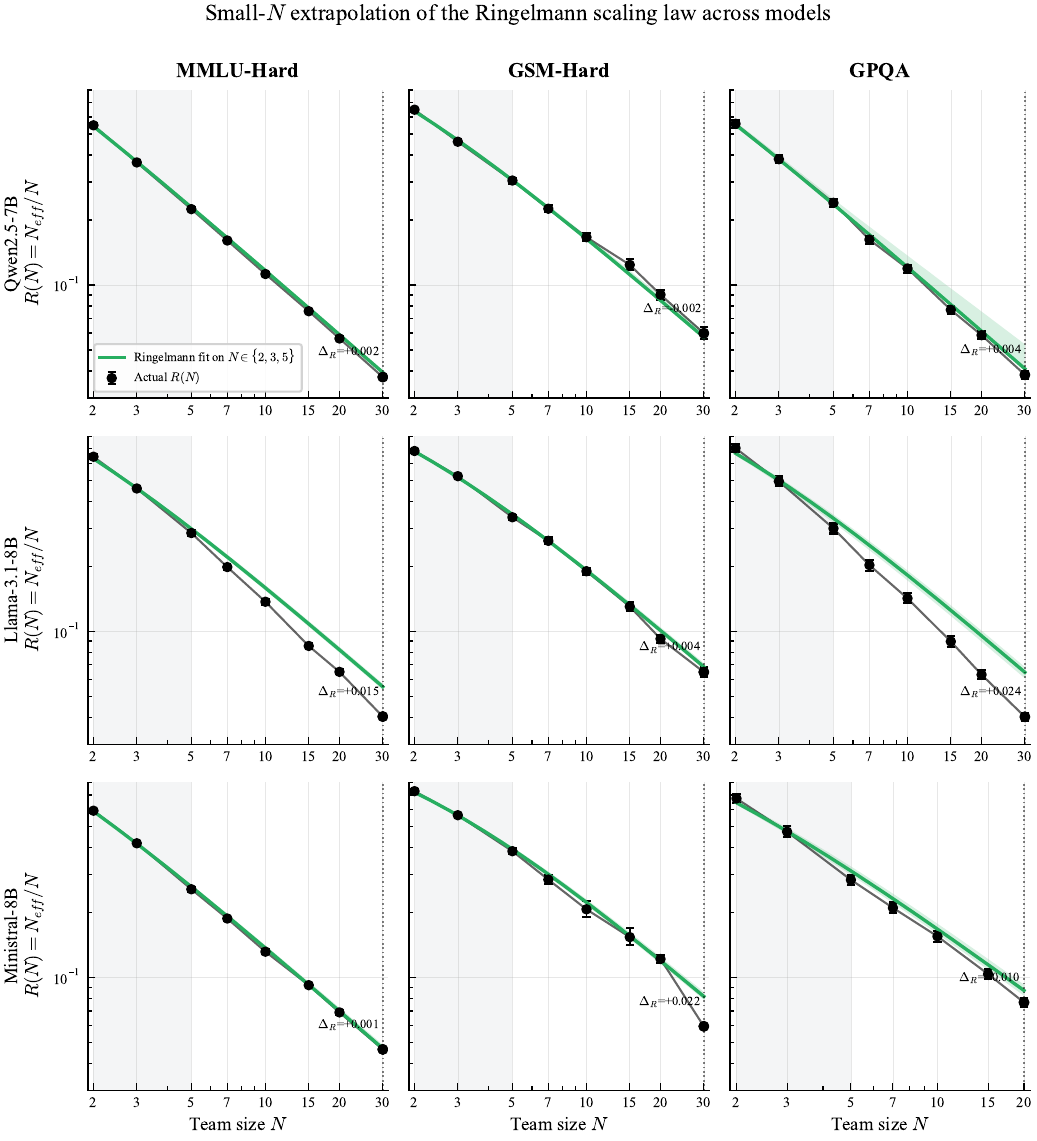}
\caption{\textbf{Small-team runs predict the large-team Ringelmann ceiling
across models and tasks.} Rows: Qwen2.5-7B (top), Llama-3.1-8B (middle), Ministral-8B (bottom). Columns: MMLU-Hard, GSM-Hard, GPQA. Black points show observed answer-diversity efficiency $R(N)=N_\text{eff}/N$ with $95\%$ item-level bootstrap intervals. Green curves are the Ringelmann form $R(N)=1/(1{+}c(N{-}1)N^{-\beta})$ estimated only on $N\in\{2,3,5\}$ (gray training band), with $95\%$ bootstrap envelope, extrapolated to the largest $N$ evaluated for that model and task. Residuals $\Delta_R$ at the largest $N$ are reported per panel. Both axes are on log scale (matching Figure~\ref{fig:hetero_extrapolation}). Cross-family comparison against alternative functional forms (Const-$\rho$ Kish, Latan\'{e} power law, logistic) appears in Table~\ref{tab:held_out_audit}, where held-out errors across all three open-weight families confirm the same ranking.}
\label{fig:efficiency_extrapolation}
\end{figure}

\begin{figure}[H]
\centering
\includegraphics[width=\textwidth]{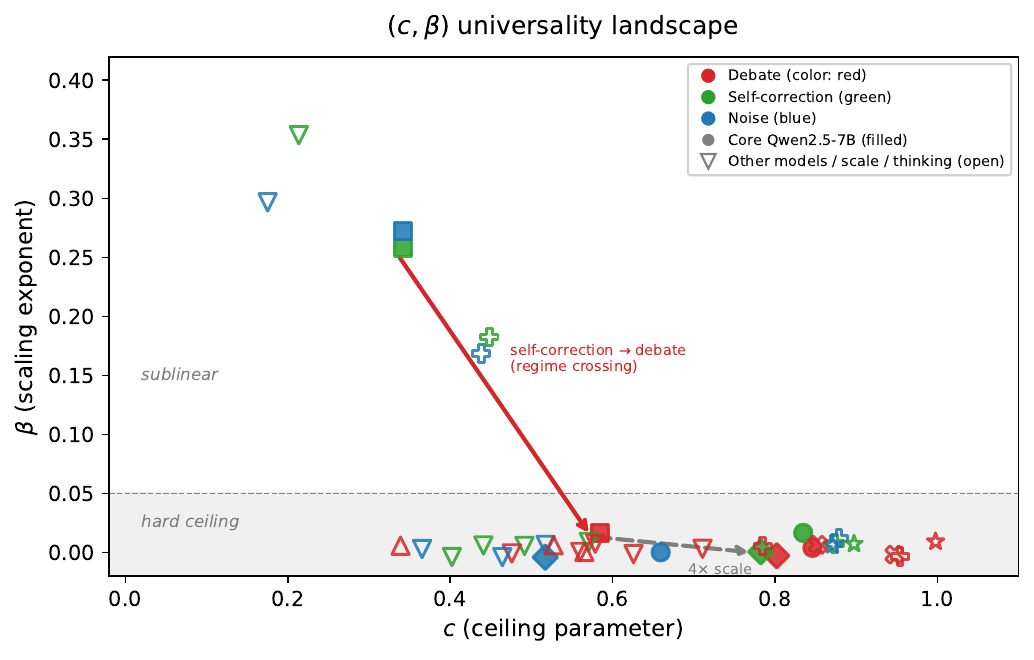}
\caption{\textbf{The $(c, \beta)$ landscape across evaluated configurations.}
Each point is one Ringelmann estimate across 6~models plus heterogeneous teams, 3~tasks, and 3~conditions (44~estimates total, also reported in Table~\ref{tab:master_fits}). \emph{Color}: condition (red = debate, green = self-correction, blue = noise). \emph{Filled shapes}: core Qwen2.5-7B (circles = MMLU-Hard, squares = GSM-Hard, diamonds = GPQA). \emph{Open shapes}: cross-model variants (inverted triangles = Llama-3.1-8B / Ministral-8B, plus = $4\times$ scale, upward triangles = heterogeneous, stars = thinking). \textbf{Red arrow}: answer-level regime crossing on GSM-Hard (self-correction sublinear $\to$ debate hard-ceiling), same pattern on Llama-3.1-8B (Table~\ref{tab:regime_comparison}). \textbf{Gray dashed arrow}: $4\times$ model scale tightens the hard ceiling on GSM debate ($c{:}\ 0.57 \to 0.77$), and the same ceiling-tightening pattern holds on MMLU debate ($0.85 \to 0.96$).}
\label{fig:universality}
\end{figure}

\subsection{Self-correction, noise, and self-consistency: detailed analysis}
\label{app:noise_detail}

This subsection supports C4 (initial correlation $\rho^{(0)}$ as the lever) by examining three sub-conditions in detail. The shared-form result of \S~\ref{sec:universality} rests on these sub-conditions.

\textbf{Self-correction (no peers, same law).} Self-correction uses the same revision prompt and model population as debate but removes all peer information. On GSM-Hard it produces $c = 0.34$, $\beta = 0.25$ (sublinear regime), with different parameters from debate but the same functional form and the same quality of fit. The shift is interpretable: without peer-induced conformity, correlation is lower and answer diversity decays more slowly with $N$, placing self-correction in a qualitatively better scaling regime.

\textbf{Noise placebo (random peers, same law).} The noise placebo replaces true peer answers with randomly shuffled text from other items. It produces $c = 0.33$, $\beta = 0.27$ on GSM-Hard, statistically indistinguishable from self-correction ($\widehat{N}_\text{eff}$ within $0.4$ at every team size). Even though the law is derived under explicit peer-conformity dynamics, it fits the no-peer placebo with matching parameters. The law therefore captures a structural property of correlated LLM outputs under revision, not a property of peer debate specifically.

\textbf{Self-consistency (no revision, predicted from $\rho^{(0)}$).} Self-consistency, majority vote over round-1 answers without any revision or communication, is the true zero-interaction baseline. The framework predicts its scaling directly from the pre-communication correlation $\rho^{(0)}$. Measuring round-1 pairwise agreement gives $\rho^{(0)} \approx 0.97$ (MMLU-Hard) and $\rho^{(0)} \approx 0.51$ (GSM-Hard). On MMLU-Hard, fitting the Ringelmann curve to round-1 answers recovers $c = 0.93$, close to $\rho^{(0)}$.\footnote{Round-1 $\rho^{(0)} = 0.97$ is slightly higher than the post-debate $\rho = 0.94$ in Table~\ref{tab:c_vs_rho} because debate redistributes some answers on MMLU-Hard, breaking prior consensus on a small fraction of items.} On GSM-Hard, self-consistency lands in the sublinear answer-level regime ($c = 0.43$, $\beta = 0.25$, $R^2 = 0.9997$), the same regime as self-correction and noise. Accuracy is flat from $N{=}2$ to $N{=}30$: low solo accuracy ($p \approx 0.10$) prevents the correct answer from winning a plurality even as answer-space entropy grows.

\subsection{Robustness of the functional form}
\label{app:fit_robustness}

Three independent lines of evidence argue against an overfitting reading.

\noindent\emph{(i)~Out-of-sample extrapolation.}
Parameters estimated on $N \le 5$ predict $N{=}30$ across the evaluated configurations (Figure~\ref{fig:efficiency_extrapolation}). Out-of-sample extrapolation is the standard test for a power-law fit versus a flexible curve, formalized in the held-out audit (Table~\ref{tab:held_out_audit}).

\noindent\emph{(ii)~Interpretable parameter shifts.}
The shifts across conditions follow the mean-field prediction: peer debate increases $c$ and lowers $\beta$ (more conformity, tighter ceiling), noise and self-correction overlap (no peer signal, same re-evaluation effect), and self-consistency is predicted parametrically from $\rho^{(0)}$ with no fitting at all.

\noindent\emph{(iii)~Cross-validation against an external metric.}
The condition ordering in estimated $c$ (debate $>$ self-corr.\ $\approx$ noise) is confirmed by the independently computed correctness-level ICC $\rho_\text{ICC}^{\text{corr}}$ (Table~\ref{tab:c_vs_rho}), which does not enter the fitting procedure. The ordering survives item-difficulty conditioning: $\rho_\text{cond}^{\text{corr}} = 0.74$ (debate) $> 0.33$ (self-corr.)\ $> 0.04$ (noise) on MMLU-Hard (App.~\ref{app:corr_level}). On MCQA, $c \approx \rho_\text{eff}^{\text{ans}} \approx \rho_\text{ICC}^{\text{corr}}$. On GSM-Hard the three diverge because distinct wrong answers inflate answer entropy, but the ordering holds.

\subsection{Across 7--8B model families}

\textbf{Regime crossing on GSM-Hard replicates across families.} The main-text canonical case (Qwen2.5-7B, §\ref{sec:universality}) shows debate at $c = 0.57$, $\beta \approx 0$ (hard ceiling) versus self-correction at $c = 0.34$, $\beta = 0.25$ and noise at $c = 0.33$, $\beta = 0.27$ (both sublinear). On Llama-3.1-8B the same crossing is visible with sharper contrast: debate $c = 0.47$, $\beta = 0.000$, self-correction $c = 0.22$, $\beta = 0.361$, noise $c = 0.18$, $\beta = 0.300$ (Table~\ref{tab:master_fits}). On Ministral-8B, debate sits at $c = 0.38$, $\beta \approx 0$ (hard ceiling), while self-correction and noise are sublinear ($c = 0.29$, $\beta = 0.24$ and $c = 0.32$, $\beta = 0.17$, both fitted on $N \le 20$). All three families exhibit the regime crossing: hard-ceiling debate, sublinear when peer content is removed.

\textbf{Hard ceiling on MCQA across families.} All three 7--8B families produce $\beta = 0$ on MCQA (Table~\ref{tab:cross_model}, top block). Cross-model variation appears in $c$ (Llama $0.62$ vs.\ Qwen $0.85$ vs.\ Ministral $0.71$) rather than in the regime. The hard-ceiling regime on GPQA, originally established on Qwen, also replicates on Llama-3.1-8B ($\beta=0$, $c=0.55$) and on Ministral-8B ($\beta=0$, $c=0.57$, with self-correction $c=0.41$ and noise $c=0.37$ also at $\beta=0$), giving three-family confirmation that GPQA tracks MMLU rather than GSM. The lower $c$ for the smaller models on GSM debate is consistent with more diverse errors from weaker solo accuracy.

\textbf{Cross-model parameter consistency.} Re-running the parameter-estimation audit (Table~\ref{tab:held_out_audit}) and the modal/non-modal transition diagnostics (Table~\ref{tab:transitions}) on Llama-3.1-8B preserves both rankings: Ringelmann remains the lowest held-out error family ($11.9\%$ mean relative error vs $14.1\%$ for const-$\rho$ Kish, $40.8\%$ for power law, $69.7\%$ for logistic), and the qualitative transition pattern matches Qwen (modal-stay high on MCQA, lower on GSM-Hard, debate non-modal-to-modal rate $2$--$3{\times}$ the rate under self-correction or noise, per-peer $\alpha_l$ decaying with $N$). Llama's modal-stay rate on GSM is lower than Qwen's ($26$--$40\%$ vs $54$--$63\%$), even though Llama's GSM solo accuracy ($\sim 17\%$) is slightly higher than Qwen's ($10.5\%$). ``Less capable'' is therefore not a complete account of the parameter differences. The framework's prediction is consistent with what we observe: regimes are shared (both at $\beta \approx 0$ on GSM debate) while parameters $c$ and $\alpha_s$ differ in the predicted direction (lower $c$ correlates with higher modal leakage). Cross-model parameter shifts without regime crossing are exactly what \S~\ref{sec:framework} predicts.

\subsection{Scale sensitivity: 7B vs.\ 32B}

A key question is whether the hard-ceiling regime ($\beta \approx 0$) persists at larger model scale. Qwen2.5-32B-Instruct provides a within-family evaluation at $4\times$ parameter scale on GSM-Hard, MMLU-Hard, and GPQA (200 / 200 / 198 items, team sizes $N{=}2$--$30$, all three conditions).

Three findings emerge. First, \textbf{32B debate tightens the hard ceiling}. At 7B, GSM-Hard debate has $c = 0.57$, $\beta \approx 0$. At 32B, $c$ rises to $0.77$ ($\beta = 0.00$, $R^2 = 0.998$), compressing the diversity ceiling from $1/c = 1.75$ to $1.30$. The same tightening shows on GPQA: $c$ moves from $0.81$ at 7B to $0.97$ at 32B (debate, $\beta \approx 0$, Table~\ref{tab:master_fits}), placing the 32B GPQA ceiling within $1.03$ effective agents at $N{=}30$. Second, \textbf{the 32B model is more correlated, not less}: $\widehat{N}_\text{eff}(30) = 1.24$ at 32B versus $1.79$ at 7B on GSM-Hard, and $1.03$ versus $1.15$ on GPQA. Solo accuracy rises but this gain is driven by the smarter base model, not by better collective reasoning. Third, \textbf{noise $\approx$ self-correction replicates at 32B with $N$ up to $30$}: self-correction gives $c = 0.45$, $\beta = 0.18$ and noise gives $c = 0.45$, $\beta = 0.18$, nearly identical, confirming the mechanism-agnosticity of the Ringelmann law across model scales. On GPQA the noise condition has partial item coverage at $N{=}30$ ($122/198$), but the $(c, \beta)$ fit is stable on the full $N{=}2$--$30$ grid.

\section{Communication density ablation}
\label{app:k_sweep}

This appendix supports C2 (the $k\tau$ product theorem) by ablating communication density. The main experiments use fully connected teams ($k = N{-}1$, every agent sees all peers). A natural question is whether conformity can be reduced by limiting it. We run a dedicated $k\tau$ ablation on GSM-Hard (200 items, Qwen2.5-7B) with $k \in \{1, 2, 4, 9, 29\}$ peers, $\tau \in \{0, \ldots, 5\}$ communication rounds, and $N \in \{3, 10, 30\}$.

\begin{table}[H]
\caption{Scaling exponent $\beta$ by communication density $k$
(Qwen2.5-7B, 200 items, final round $\tau{=}5$ for GSM-Hard, $\tau{=}6$ for MMLU-Hard). $\beta_\text{ent}$: entropy scale ($N_\text{eff}^{\text{ans}} = 2^H$). $\beta_\text{agr}$: pairwise agreement scale ($N_\text{eff}^{\text{agr}} = N/(1{+}(N{-}1)\rho)$). The theorem operates on $\rho$. On free-form GSM-Hard the two scales diverge (large $\beta_\text{ent}$ from wrong-answer mass spread, \S~\ref{sec:framework}). On MCQA MMLU-Hard the two scales coincide because the answer space is bounded ($K{=}4$).}
\label{tab:k_sweep}
\centering
\small
\begin{tabular}{lcccc}
\toprule
$k$ & GSM $\beta_\text{agr}$ & GSM $\beta_\text{ent}$ &
      MMLU $\beta_\text{agr}$ & MMLU $\beta_\text{ent}$ \\
\midrule
 1  & 0.17 & 0.39 & 0.06 & 0.11 \\
 2  & 0.14 & 0.31 & 0.06 & 0.10 \\
 4  & 0.11$^*$ & 0.33 & 0.02$^*$ & 0.05 \\
 9  & 0.09$^*$ & 0.25 & 0.03$^*$ & 0.05 \\
\bottomrule
\multicolumn{5}{l}{\footnotesize $^*$2-point log-slope estimate.}
\end{tabular}
\end{table}

Table~\ref{tab:k_sweep} reveals a scale-dependent picture and confirms the theorem on a second task. On the \emph{pairwise agreement} scale (the theorem's native variable), $\beta_\text{agr}$ is small on both tasks: $0.09$--$0.17$ on GSM-Hard and $0.02$--$0.06$ on MMLU-Hard, decreasing monotonically with $k$ and converging toward the theorem's prediction of $\beta = 0$. On the \emph{entropy} scale, the two tasks separate sharply. On bounded-answer-space MMLU-Hard ($K{=}4$ choices), $\beta_\text{ent} \approx \beta_\text{agr}$: the two scales coincide, as predicted in \S~\ref{sec:framework}. On free-form GSM-Hard, $\beta_\text{ent}$ is much larger ($0.25$--$0.39$) because distinct wrong numeric answers inflate answer entropy $H(\mathbf{a})$ without proportionally reducing pairwise agreement, and at fixed $k$ this inflation is amplified because minority answers spread across more agents in larger teams. The theorem is thus confirmed on its native scale across both task types. The GSM--MMLU contrast in $\beta_\text{ent}$ visualizes the wrong-answer-mass-spread phenomenon as a task-structural property, not a violation of the theorem.

\textbf{The $k\tau$ product structure.} The matched-$k\tau$ collapse test is reported in Table~\ref{tab:ktau_collapse} of \S\ref{sec:ktau}. The aggregate pattern is consistent across both tasks: spread stays at $\Delta\rho \approx 0.04$--$0.08$ on GSM-Hard and contracts on MMLU-Hard from $\Delta\rho \approx 0.13$ at $k\tau{=}4$ to $\Delta\rho \approx 0.015$ at $k\tau{=}12$. The looseness at low $k\tau$ is driven by single-round configurations $(k, 1)$, where the theorem's binary-coarsening derivation introduces second-order corrections that its linearization does not capture; restricting to $\tau \ge 2$ tightens the spread on both tasks.

\textbf{Stratifying by initial agreement.} The aggregate $\Delta\rho$ in Table~\ref{tab:ktau_collapse} averages over all items per cell. Conditioning instead on per-item initial agreement $\bar\rho^{(0)}$ (a difficulty proxy) reveals a residual structure on GSM-Hard ($N{=}10$): restricting to fair multi-round comparisons ($\tau \ge 2$), median $\Delta\rho$ across realizations of the same $k\tau$ is $0.027$ in high-agreement items ($\bar\rho^{(0)} > 0.7$, $80$ items), $0.022$ in mid-agreement items ($\bar\rho^{(0)} \in [0.4, 0.7]$, $40$ items), and $0.079$ in low-agreement items ($\bar\rho^{(0)} < 0.4$, $80$ items, max $\Delta\rho = 0.10$). The product law thus holds tightly where the high-agreement linearization of Thm.~\ref{thm:comm-density} is valid and loosens by roughly $3\times$ where it is least valid. Single-round ($\tau{=}1$) realizations also deviate from multi-round outcomes at the same product, consistent with larger second-order corrections during the $\tau{=}1$ initial transient. Figure~\ref{fig:ktau_stratified} visualizes the three regimes.

\begin{figure}[H]
\centering
\includegraphics[width=\textwidth]{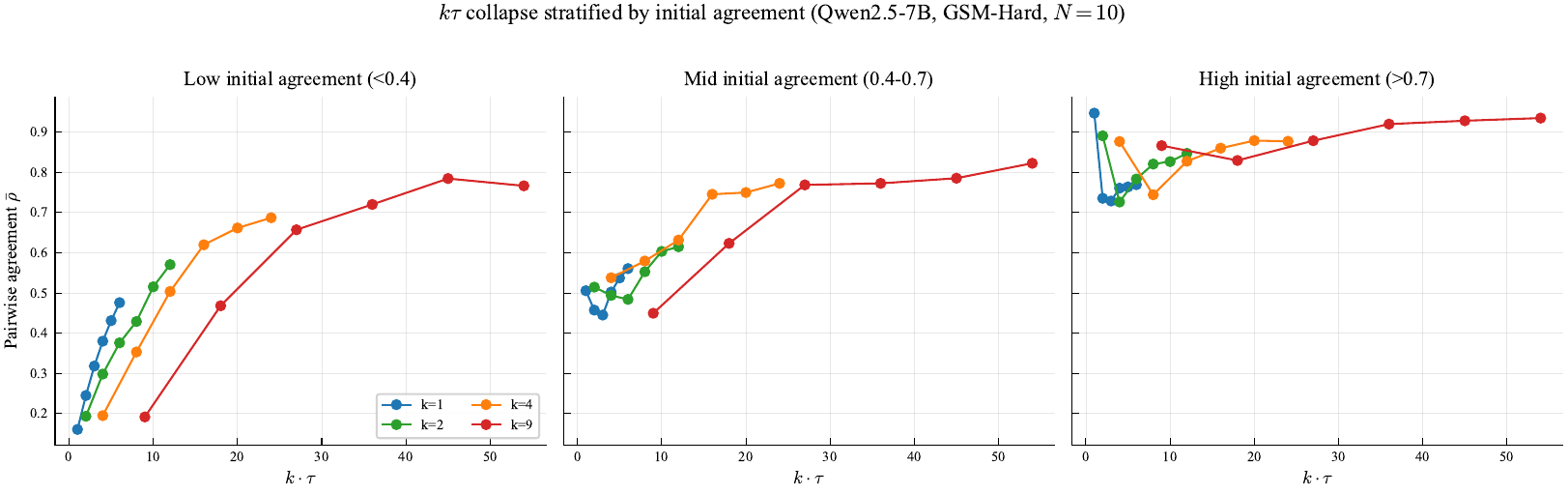}
\caption{$k\tau$ collapse stratified by initial item-level agreement
(Qwen2.5-7B, GSM-Hard, $N{=}10$). Curves are pairwise agreement $\bar\rho$ vs.\ $k\tau$, one per peer count $k$. Right panel (high agreement, $\bar\rho^{(0)} > 0.7$): curves nearly coincide, consistent with the theorem's high-agreement regime. Left panel (low agreement, $\bar\rho^{(0)} < 0.4$): visible cross-$k$ separation at small $k\tau$, narrowing as $k\tau$ grows.}
\label{fig:ktau_stratified}
\end{figure}

\textbf{Why the large $\beta_\text{ent}$ does not violate the theorem.} The theorem predicts contraction of $\rho$, not entropy. At dense communication ($k \to N{-}1$), $\beta_\text{agr} \to 0$ on both tasks while $\beta_\text{ent}$ persists only on GSM-Hard, isolating answer-space inflation as a task-structural quantity independent of communication dynamics.

\section{Two-rate mean-field extension}
\label{app:two_rate}

This appendix supports C2 by extending the mean-field model with asymmetric conformity rates and grounding assumption (A2) in agent-level transition diagnostics. The single conformity coefficient $\alpha$ in Eq.~\ref{eq:rho} averages over two distinct processes with empirically different magnitudes.

\begin{definition}[Asymmetric Conformity Rates]
\label{def:two_rate}
Let $\alpha_s := P(\text{correct agent adopts wrong peer majority})$ denote the \emph{sycophancy rate} and $\alpha_l := P(\text{wrong agent adopts correct peer majority})$ denote the \emph{learning rate}.
\end{definition}

Table~\ref{tab:two_rate} reports these rates from the agent-level debate data. On both tasks, $\alpha_l > \alpha_s$: agents are more readily corrected than corrupted. The asymmetry is moderate on MMLU-Hard ($\alpha_s/\alpha_l \approx 0.67$) and pronounced on GSM-Hard ($\approx 0.30$), consistent with the stronger persuasive signal carried by correct mathematical reasoning.

\begin{table}[H]
\caption{Asymmetric conformity rates at $N{=}10$, round~3 (Qwen2.5-7B, debate).}
\label{tab:two_rate}
\centering
\small
\begin{tabular}{lcccc}
\toprule
Task & $\alpha_s$ ($n$) & $\alpha_l$ ($n$) & Ratio & Unconditional C$\to$W / W$\to$C \\
\midrule
MMLU-Hard & 0.60 (119) & 0.89 (141) & 0.67 & 0.025 / 0.045 \\
GSM-Hard  & 0.24 (170) & 0.79 (233) & 0.30 & 0.005 / 0.088 \\
\bottomrule
\end{tabular}
\end{table}

\textbf{Accuracy dynamics (item-conditioned heuristic).} When the majority is correct, wrong-minority agents switch at rate~$\alpha_l$, and when the majority is wrong, correct-minority agents switch at rate~$\alpha_s$:
\begin{equation}
p^{(\tau+1)} \approx
\begin{cases}
p^{(\tau)} + (1 - p^{(\tau)})\,\alpha_l & \text{if } p^{(\tau)} > 0.5, \\[4pt]
p^{(\tau)}(1 - \alpha_s) & \text{if } p^{(\tau)} < 0.5.
\end{cases}
\label{eq:accuracy_two_rate}
\end{equation}
When $\alpha_l > \alpha_s$, accuracy grows toward~1 faster than it shrinks toward~0.

\textbf{Agreement dynamics and robustness of scaling predictions.} On a per-item basis, items where the modal answer is correct contract pairwise disagreement at rate $\alpha_l$ (non-modal agents adopting correct modal), and items where the modal answer is wrong contract at rate $\alpha_s$ (correct minority adopting wrong modal). Let $f_c$ denote the fraction of items where the modal answer is correct (the population's modal-level accuracy). The population-mean disagreement contracts as
\[
d^{(\tau+1)} \approx d^{(\tau)}(1 - \alpha_\text{eff}), \qquad \alpha_\text{eff} = f_c\,\alpha_l + (1{-}f_c)\,\alpha_s.
\]
As long as $\alpha_\text{eff} > 0$ (which holds whenever $\alpha_l, \alpha_s \in (0,1)$), the contraction is geometric, so $\rho^{(\tau)} \to 1$ as $\tau \to \infty$ and the hard-ceiling prediction is preserved. The $k\tau$ product structure extends by the same argument as Thm.~\ref{thm:comm-density}: replacing the single $\alpha_1$ with $\alpha_\text{eff}$ in Eq.~\ref{eq:contraction} gives the same $k\tau$-product collapse with a renormalized contraction rate. The single-$\alpha$ model is therefore sufficient for \emph{agreement dynamics and scaling predictions}, while the two-rate extension is needed for \emph{accuracy-trajectory predictions}, where the asymmetry $\alpha_l > \alpha_s$ drives the convergence of $f_c \to 1$.

\textbf{Empirical transition matrices grounding Thm.~\ref{thm:comm-density}.} Thm.~\ref{thm:comm-density} idealizes the modal class as stable (assumption A2). Table~\ref{tab:transitions} reports the empirical modal/non-modal transition rates on Qwen2.5-7B across all three tasks and conditions, computed by anchoring on the round-1 modal answer per item and tracking persistence to round~3. Three observations: (i)~Modal-stay rates are high on MCQA ($85$--$91\%$) and moderate on GSM-Hard ($54$--$63\%$), placing the modal-leak rate ($1 - \text{modal-stay}$) in $[0.09, 0.46]$ and confirming that A2 is well-satisfied on bounded-answer tasks and approximate on free-form math. (ii)~Non-modal-to-modal rates are substantially higher under debate ($38$--$62\%$) than under self-correction or noise ($10$--$25\%$), consistent with the conformity mechanism. (iii)~The per-peer learning rate $\alpha_l := P(\text{nm}\to\text{m}) / k$ shrinks with $N$ on every cell, consistent with the high-agreement saturation argument: as the modal class grows, the marginal information from each additional modal peer decays, and the contraction enters the plateau regime predicted by Eq.~\ref{eq:contraction}.

\begin{table}[H]
\caption{Empirical modal/non-modal transition rates on Qwen2.5-7B
(round~1 $\to$ round~3, anchored on round-1 modal answer per item). Modal-stay $:= P(\text{R1-modal stays at the same answer through R3})$. The modal-leak rate is $1 - \text{modal-stay}$. nm$\to$m $:=$ probability that an R1-non-modal agent moves to the R1-modal answer by R3. $\alpha_l := $ nm$\to$m $/ k$ (per-peer learning rate).}
\label{tab:transitions}
\centering
\small
\begin{tabular}{ll|ccc|ccc|ccc}
\toprule
& & \multicolumn{3}{c|}{Modal-stay} & \multicolumn{3}{c|}{nm$\to$m} & \multicolumn{3}{c}{$\alpha_l$ per peer} \\
Task & Cond. & $N{=}5$ & $N{=}10$ & $N{=}30$ & $N{=}5$ & $N{=}10$ & $N{=}30$ & $N{=}5$ & $N{=}10$ & $N{=}30$ \\
\midrule
MMLU & deb. & .85 & .88 & .90 & .38 & .53 & .62 & .096 & .059 & .021 \\
     & self & .89 & .90 & .90 & .18 & .24 & .24 & .045 & .027 & .008 \\
     & noi. & .83 & .84 & .88 & .15 & .28 & .19 & .038 & .031 & .007 \\
\midrule
GSM  & deb. & .54 & .59 & .63 & .25 & .34 & .41 & .062 & .038 & .014 \\
     & self & .59 & .60 & .60 & .02 & .03 & .03 & .006 & .003 & .001 \\
     & noi. & .63 & .68 & .68 & .02 & .02 & .02 & .004 & .002 & .001 \\
\midrule
GPQA & deb. & .88 & .90 & .91 & .52 & .45 & .55 & .130 & .050 & .019 \\
     & self & .88 & .88 & .87 & .17 & .25 & .20 & .044 & .028 & .007 \\
     & noi. & .79 & .79 & .80 & .11 & .21 & .11 & .027 & .024 & .004 \\
\bottomrule
\end{tabular}
\end{table}

These diagnostics support A2 as an approximation rather than an exact identity: the modal class is leaky at non-trivial rates, especially on GSM-Hard, and the robustness analysis above shows that the $k\tau$ structure survives this leakiness with a renormalized contraction rate.

\section{Deployment: cost, throughput, and accuracy projection}
\label{app:deployment}

This appendix covers the deployment-facing implications of the scaling law: the small-$N$ calibration procedure used to estimate $(c, \beta)$ on a new model or task (\S\ref{app:measurement}), token-cost scaling under dense communication (\S\ref{app:cost}), a heuristic accuracy projection from small-team calibration (\S\ref{app:accuracy_projection}), and the detailed economic argument behind the deployment claims in \S~\ref{sec:results} (\S\ref{app:deployment_econ}).

\subsection{Measurement procedure: estimating \texorpdfstring{$(c, \beta)$}{(c, beta)} on a new model}
\label{app:measurement}

The recipe below is what the released script \texttt{analyses/held\_out\_audit.py} implements. It calibrates $(c, \beta)$ from a small-team pilot and reports the predicted extrapolation to a target $N$.

\textbf{Algorithm.}
\begin{enumerate}[leftmargin=1.4em, itemsep=2pt, topsep=2pt]
\item \emph{Pick a model and task.} Pick a small-$N$ grid. The paper
uses $\{2, 3, 5\}$ for the calibration set and a held-out target $N \in \{7, 10, 15, 20, 30\}$.
\item \emph{Run teams.} For each $N$, draw $M$ items (we use
$M{=}200$--$1017$). For each item, sample $N$ independent agent responses (single-shot ensembling) or run $\tau$ rounds of debate with $k = N{-}1$ peers, recording the final answer of each agent.
\item \emph{Per-team effective count.} For each item $i$, compute the
Shannon entropy $H_i$ in bits over the $N$ agent answers and set $N_\text{eff}^{\text{ans}}(i) = 2^{H_i}$ (\S~\ref{sec:framework}).
\item \emph{Aggregate to $R(N)$.} Average across items at each $N$:
$R(N) = \frac{1}{M} \sum_i N_\text{eff}^{\text{ans}}(i) / N$.
\item \emph{Fit $(c, \beta)$.} Nonlinear least squares of
$R(N) = 1 / (1 + c(N{-}1) N^{-\beta})$ on the calibration grid, with bounds $c \in [0, 20]$, $\beta \in [0, 1]$. The $c$ bound is permissive for optimizer stability. The theoretical upper bound from $\rho_N \le 1$ at $N=2$ is $c \le 2^\beta \le 2$, and all $44$ fits in this paper land at $c \le 1$.
\item \emph{Read the regime.} $\beta \approx 0$ is hard ceiling at
$1/c$, $0 < \beta < 1$ is sublinear scaling ($N_\text{eff} \sim N^\beta / c$), and $\beta \approx 1$ is linear, discounted by $1/(1{+}c)$.
\end{enumerate}

\textbf{Choice of $N_\text{eff}$ scale.} On bounded-answer (MCQA) tasks, $N_\text{eff}^{\text{ans}} = 2^H$ is tight: the entropy $H \le \log_2 K$ saturates and matches the agreement-based $N_\text{eff}^{\text{agr}} = N / (1 + (N{-}1)\rho)$ to within a few percent. On free-form tasks (GSM-Hard), distinct wrong numeric strings inflate $H$ without reducing $\rho$, so $N_\text{eff}^{\text{ans}}$ is an upper bound on the effective independent count rather than a tight measure. The theorem-native variable is $\rho$. For free-form tasks, replace Step~3 with: per item, compute pairwise agreement $\rho_i = \frac{2}{N(N{-}1)} \sum_{j<k} \mathbb{1}[a_j = a_k]$ and use $N_\text{eff}^{\text{agr}}(i) = N / (1 + (N{-}1)\rho_i)$. The two scales coincide on MCQA and decouple on free-form math, and the paper reports both, using $\rho$ for the theorem (see App.~\ref{app:k_sweep}).

\textbf{Worked example: Qwen2.5-7B on GSM-Hard, debate condition.} With $M{=}1017$ items per $N$ and $\tau{=}3$:
\begin{itemize}[leftmargin=1.4em, itemsep=0pt, topsep=2pt]
\item Step 4 outputs (mean $N_\text{eff}^{\text{ans}}$ per $N$):
$R(2) = 0.93$, $R(3) = 0.85$, $R(5) = 0.65$.
\item Step 5 yields $c = 0.57$, $\beta = 0.014$, $R^2 = 0.999$.
\item Step 6: $\beta \approx 0$, hard-ceiling regime, predicted ceiling
$1/c = 1.74$.
\item Held-out check at $N{=}30$: predicted $N_\text{eff} = 1.72$, observed $1.79$, relative error $4\%$. Across all nine (task, condition) cells, the mean relative error is $5.5\%$ (Table~\ref{tab:held_out_audit}, App.~\ref{app:cross_model}).
\end{itemize}

\textbf{Practical defaults.} $M \ge 200$ items per $N$ is sufficient for stable means. Calibration grid $\{2, 3, 5\}$ minimizes compute while keeping the fit identified (three points for two parameters), and larger grids $\{2, 3, 5, 7, 10\}$ tighten the fit at modest extra cost. The same procedure fits $(c, \beta)$ on the answer-diversity ($N_\text{eff}^{\text{ans}}$) and correctness-redundancy ($N_\text{eff}^{\text{corr}}$) levels separately. The two coincide on bounded-answer MCQA and decouple on free-form math (\S~\ref{sec:framework}).

\subsection{Token cost scaling}
\label{app:cost}

The Ringelmann ceiling translates directly into wasted compute. We report cost as prompt tokens per effective agent ($\text{Prompt}/\widehat{N}_\text{eff}$), a hardware-independent measure that mirrors per-token API pricing and isolates the scaling claim from batch-size and utilization effects. At $N{=}30$, debate spends roughly $91$K prompt tokens per effective agent on MMLU-Hard and $64$K on GSM-Hard, versus $7.5$K and $1.4$K for self-correction, a $12{\times}$ and $46{\times}$ ratio. On GPQA the gap is similar ($101$K vs.\ $10$K, ${\approx}10{\times}$). Output-token cost follows the same pattern but is dominated by the prompt term in debate, where each agent must read $N{-}1$ peer messages. Figure~\ref{fig:prompt_scaling} shows the super-linear growth of debate prompt tokens versus the linear growth of self-correction across all three tasks, which drives the cost asymmetry.

\begin{table}[H]
\caption{Scaling summary at $N{=}30$ (Qwen2.5-7B).
$\Delta$ acc: team accuracy at $N{=}30$ minus solo accuracy at $N{=}1$ (the round-1 ensemble baseline; $0.555$ on MMLU-Hard, $0.100$ on GSM-Hard, $0.303$ on GPQA). Prompt/$N_\text{eff}$: prompt tokens per effective agent, quantifying redundant compute.}
\label{tab:scaling}
\centering
\small
\begin{tabular}{llcccc}
\toprule
Task & Condition & Acc & $\Delta$ acc & $N_\text{eff}$ & Prompt/$N_\text{eff}$ \\
\midrule
\multirow{2}{*}{MMLU} & Debate & .579 & +.023 & 1.12 & 91K \\
                       & Self-corr. & .587 & +.032 & 1.21 & 7.5K \\
\midrule
\multirow{2}{*}{GSM} & Debate & .264 & +.164 & 1.79 & 64K \\
                      & Self-corr. & .282 & +.182 & 6.10 & 1.4K \\
\midrule
\multirow{2}{*}{GPQA} & Debate & .318 & +.015 & 1.15 & 101K \\
                       & Self-corr. & .313 & +.010 & 1.29 & 10K \\
\bottomrule
\end{tabular}
\end{table}

\begin{figure}[H]
\centering
\includegraphics[width=0.5\textwidth]{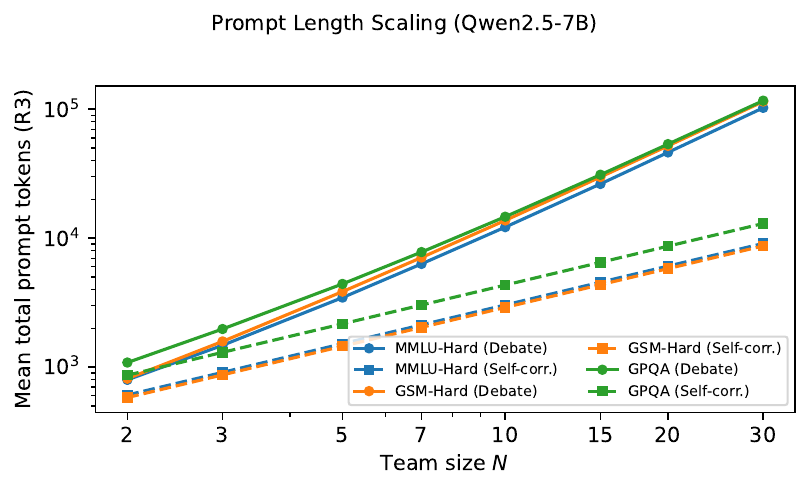}
\caption{Total prompt tokens per item at R3 (Qwen2.5-7B), across three tasks and two conditions. Solid lines: debate. Dashed lines: self-correction. Debate prompt cost grows super-linearly in $N$ because each agent reads $N{-}1$ peer messages, while self-correction grows linearly. The resulting $\sim 10$--$45\times$ gap at $N{=}30$ reflects this scaling difference.}
\label{fig:prompt_scaling}
\end{figure}

\subsection{Heuristic accuracy projection}
\label{app:accuracy_projection}

\textbf{Solo accuracy and the Condorcet bridge.} On a separate axis from the scaling regime classification, solo accuracy $p$ interacts with majority voting through the Condorcet CDF. When the model is already strong ($p > 0.5$, MMLU-Hard), correlation prevents the majority-vote gain from materializing: even a small $\rho > 0$ collapses $N_\text{eff}$ near $1$, and accuracy gains over solo are at most a few points. When individual accuracy is low ($p \approx 0.10$--$0.30$, GSM-Hard and GPQA), absolute accuracy still improves substantially from $10\%$ to $26$--$28\%$ on GSM-Hard via re-evaluation rather than peer information (self-correction $+18.2$~pp, debate $+16.4$~pp). On GSM-Hard, this gain is possible despite $p \ll 0.5$ because wrong answers scatter across many distinct numeric strings, so the correct answer wins by plurality even at low individual accuracy: this is the regime where free-form math diverges from the binary Condorcet threshold.

The calibration protocol in \S~\ref{sec:results} predicts the structural ceiling $N_\text{eff}$, which is the paper's primary quantitative prediction. As a heuristic, one can plug the global mean individual accuracy $\bar{p}$ into the correlated-Condorcet approximation: $\text{PP}(N_\text{eff}, \bar{p})$.

\textbf{Theoretical limitation.} Because $\text{PP}(\cdot, p)$ is nonlinear in $p$, applying it at the population mean $\bar{p}$ ignores item-difficulty heterogeneity: $\text{PP}(N_\text{eff},\, \mathbb{E}[p_i]) \neq \mathbb{E}[\text{PP}(N_\text{eff},\, p_i)]$. The heuristic is best understood as a practical projection layered on top of the structural ceiling law.

\textbf{$\rho$ quantities.} See \S~\ref{sec:framework} for the full notation guide. On MCQA tasks all three measures track each other (spread 0.05--0.10 for debate and self-correction, up to 0.15 for noise, Table~\ref{tab:c_vs_rho}). On free-form GSM-Hard they diverge substantially: $\rho_\text{ICC}^{\text{corr}} \approx 0.7$--$0.9$ while $\rho_\text{eff}^{\text{ans}} \approx 0.1$--$0.5$, because many distinct wrong numeric answers inflate answer entropy without contributing useful diversity. The scaling law describes the answer-level quantity, and the correctness-level ICC confirms the condition ordering (debate $>$ self-correction $\approx$ noise) independently.

\textbf{Prediction recipe (MCQA tasks).} On MCQA tasks where $\rho_\text{eff}^{\text{ans}} \approx \rho_\text{ICC}^{\text{corr}}$, the scaling law supports a direct accuracy projection. Calibrate two parameters from small teams ($N \le 5$): $\bar{p}$ (mean solo accuracy) and $\rho$ (answer-level correlation, Eq.~\ref{eq:neff}). For any target team size $N$, compute $N_\text{eff} = N / (1 + (N{-}1)\rho)$ and predict team accuracy as $\text{PP}(N_\text{eff},\, \bar{p})$. On free-form tasks where the answer-level and correctness-level correlations diverge, the Condorcet bridge becomes approximate. The scaling law still describes answer-diversity dynamics but should not be chained to accuracy via $\text{PP}$ without additional calibration.

\textbf{Empirical validation.} Per-agent accuracy is nearly constant across team sizes ($\pm 1.3$~pp across $N{=}2$--$30$ on all tasks), so the two-parameter calibration is stable. For homogeneous debate, projecting from $N \le 5$ to $N{=}7$--$30$ gives mean errors ${\le}1.4$~pp on all three tasks (Table~\ref{tab:accuracy_prediction}). Heterogeneous teams show larger errors (up to $7.2$~pp on GSM-Hard), modestly exceeding the per-item sampling noise of 200-item subsets ($\pm 3.5$~pp at $p \approx 0.3$). The excess error is consistent with the wider item-difficulty distribution in heterogeneous mixes: different model families struggle on different items, and the population-mean projection $\text{PP}(N_\text{eff}, \bar p)$ does not capture this within-item heterogeneity. For accuracy-critical applications on heterogeneous teams, item-level calibration would be preferable. The structural ceiling prediction ($N_\text{eff}$ as a function of $N$) remains accurate even when the accuracy projection through $\text{PP}$ degrades.

\begin{table}[H]
\caption{Heuristic accuracy projection from small-team calibration (debate).
$\rho$ calibrated from $N \le 5$ only, so $N{=}30$ predictions are fully out of sample (OOS).}
\label{tab:accuracy_prediction}
\centering
\small
\begin{tabular}{llc@{\hskip 8pt}cc@{\hskip 8pt}cc}
\toprule
& & Calibration & \multicolumn{2}{c}{Accuracy at $N{=}30$} & \multicolumn{2}{c}{OOS error (pp)} \\
\cmidrule(r){3-3} \cmidrule(lr){4-5} \cmidrule(l){6-7}
Setting & Task & $\rho$ & Predicted & Observed & Mean & Max \\
\midrule
\multirow{3}{*}{Homogeneous} & MMLU-Hard & .846 & .567 & .579 & 1.3 & 1.9 \\
                              & GSM-Hard  & .566 & .249 & .264 & 1.4 & 1.7 \\
                              & GPQA      & .796 & .331 & .318 & 1.1 & 1.9 \\
\midrule
\multirow{3}{*}{Heterogeneous} & MMLU-Hard & .560 & .421 & .445 & 3.4 & 4.8 \\
                                & GSM-Hard  & .349 & .272 & .320 & 5.9 & 7.2 \\
                                & GPQA      & .588 & .297 & .303 & 0.9 & 1.2 \\
\bottomrule
\end{tabular}
\end{table}

\subsection{Detailed deployment economics}
\label{app:deployment_econ}

The condensed argument in \S~\ref{sec:results} contrasts instruct-class and reasoning-model multi-agent debate. This appendix expands the reasoning behind the deployment claims.

\textbf{Cost arithmetic.} A 30-agent, 3-round debate costs roughly $30 \times 3 \times 1500 \times \$10/\text{M}\approx\$1.35$ per problem on a frontier reasoning model at low thinking budget (typical Tier-1 output pricing), and runs in minutes per query under serial generation. The same debate on a self-hosted 7--8B instruct model costs roughly $\$10^{-3}$ per problem (amortized hardware plus power) and finishes in tens of seconds with batched inference. The latency gap is similar: instruct-class generation runs at $0.3$--$1$~s per response on local A100/L40S hardware with vLLM, while reasoning models at low thinking budget take $3$--$15$~s per response under serial generation, giving wall-clock contrasts of one to two orders of magnitude.

\textbf{Where instruct multi-agent actually deploys.} The cost asymmetry concentrates multi-agent debate at the instruct-class tier in practice. Real-time, high-volume, on-prem, and on-device deployments preferentially run open-weight $7$--$8$B instruct models because of latency, per-query budget, or data-residency constraints. Reasoning models are typically reserved for one or two hard sub-tasks within a pipeline, not used as the bulk of fan-out.

\textbf{Throughput trade-off at high solo accuracy.} For configurations with high solo accuracy, the framework predicts $c \to 1$. If the prediction were to hold at multi-agent scale (untested in this paper), the deployment consequence would be straightforward: $N$ agents debating one problem and $N$ agents working independently on $N$ problems would achieve near-identical aggregate accuracy at $1/N$ throughput, in which case the framework's prescription becomes \emph{parallelize, not team}. The solo-accuracy threshold at which $c$ approaches $1$ is itself identifiable from a small-$N$ pilot.

\section{Thinking model: stronger reasoning tightens the ceiling}
\label{app:thinking}

This appendix supports C3 by evaluating whether models with extended internal reasoning might resist peer conformity better, escaping the hard-ceiling regime. We use Qwen3-8B in thinking mode on MMLU-Hard (724 items, $N \in \{2,3,5\}$, all three communication conditions).

The results are the opposite of the hypothesis. The thinking model has substantially higher solo accuracy ($87.2\%$ vs.\ $55.6\%$ for Qwen2.5-7B), yet debate drives correlation \emph{higher} ($\rho = 0.983$ at $N{=}5$ vs.\ $0.911$ for the non-thinking model) and effective team size \emph{lower} ($\widehat{N}_\text{eff} = 1.02$ vs.\ $1.12$). The estimated Ringelmann parameters confirm this: debate gives $c = 1.00$, $\beta \approx 0$, the hard-ceiling regime with ceiling $1/c = 1.00$ (Table~\ref{tab:thinking}).

\begin{table}[H]
\caption{Ringelmann comparison: thinking vs.\ non-thinking model (MMLU-Hard).
The thinking model has higher solo accuracy but a tighter Ringelmann ceiling under debate.}
\label{tab:thinking}
\centering
\small
\begin{tabular}{llccccc}
\toprule
Model & Condition & Solo acc. & $\rho$ ($N{=}5$) & $N_\text{eff}$ ($N{=}5$) & $c$ & Ceiling \\
\midrule
\multirow{2}{*}{Qwen2.5-7B} & Debate & .556 & .911 & 1.12 & .85 & 1.18 \\
                              & Self-corr. & .556 & .879 & 1.17 & .84 & 1.19 \\
\midrule
\multirow{2}{*}{Qwen3-8B-think} & Debate & .872 & .983 & 1.02 & 1.00 & 1.00 \\
                                  & Self-corr. & .872 & .920 & 1.11 & .89 & 1.12 \\
\bottomrule
\end{tabular}
\end{table}

The mechanism is straightforward at the agreement scale. Initial pairwise agreement is already high for both models on MMLU-Hard at $N{=}5$ ($\rho^{(0)} \approx 0.91$ for Qwen3-think, $\rho^{(0)} \approx 0.96$ for Qwen2.5-7B), but a single round of debate drives the thinking model to tighter consensus ($\rho = 0.983$ vs.\ $\rho = 0.911$ for the non-thinking model). The thinking model's stronger per-agent reasoning amplifies rather than dampens conformity once peer information is on the table, leaving less room for independent evidence. The framework's predictions hold without modification, and the practical implication is \emph{stronger models need fewer debating agents, not more}.

\section{Aggregator robustness}
\label{app:aggregator}

We re-aggregate final-round team answers from the canonical homogeneous data (Qwen2.5-7B and Llama-3.1-8B, three tasks, three communication conditions: debate, self-correction, noise) under five aggregation rules, using only already-logged columns:
\begin{itemize}[leftmargin=1em, itemsep=0pt, topsep=0pt]
\item \emph{plurality}: majority vote over agent answers (paper baseline).
\item \emph{conf-weighted}: votes weighted by self-reported confidence
(a $0$--$100$ value the model is prompted to emit).
\item \emph{prob-weighted}: votes weighted by softmax probability of the
chosen answer.
\item \emph{filtered plurality}: drop the bottom $25\%$ of votes by
softmax confidence, then plurality.
\item \emph{top-1 confident}: select the single answer with maximum
softmax probability.
\end{itemize}

Self-reported confidence is poorly calibrated against the softmax (Pearson $r = 0.03$ on Qwen2.5-7B, the model emits $\ge\!85$ for almost all answers), so \emph{conf-weighted} is essentially identical to plurality and serves as a negative control. The other aggregators differ because they use the logit signal rather than the discrete answer.

Across all $(\text{model}, \text{task}, \text{condition}, N)$ cells, alternative aggregators shift team accuracy by at most $\sim 5$~pp from plurality, with mean absolute shifts under $1.3$~pp on every task. The largest shifts come from \emph{top-1-confident}: up to $3.9$~pp on MMLU-Hard, $3.7$~pp on GSM-Hard, and $5.0$~pp on GPQA, concentrated at small $N$. On the canonical Qwen2.5-7B GSM-Hard noise cell at $N{=}30$, top-1-confident lifts accuracy by $+2.4$~pp over plurality. The re-estimated $\beta$ remains $\le 10^{-9}$ (numerical noise of zero) across debate, self-correction, and noise conditions on Qwen2.5-7B MMLU-Hard, so the regime classification is aggregator-invariant on the evaluated grid even when the accuracy shifts.

Heterogeneous teams (App.~\ref{app:hetero}) use balanced compositions: $N\in\{3,6,9,15,30\}$ split equally across Qwen, Llama, and Ministral. Under this balance, two natural alternatives to plurality become algebraically trivial:
\begin{itemize}[leftmargin=1em, itemsep=0pt, topsep=0pt]
\item \emph{block-weighted} (each model family gets total weight $1$):
collapses exactly to plurality when family counts are equal.
\item \emph{family-then-vote} (each family takes its internal plurality,
then majority across families): differs from plurality by at most $1$~pp across all $(\text{task}, k, N)$ cells.
\end{itemize}
We find no aggregator effect on heterogeneous accuracies under the balanced design. \emph{Imbalanced} hetero designs (e.g.\ a single Mistral seeded into a Qwen-dominant team) are needed to discriminate the block-exchangeable prediction (App.~\ref{app:hetero_theory}) empirically, and remain future work.

\textbf{Adjudication-based aggregators.} Aggregators that score peer rationales (debate-with-cross-examination, judge-mediated consensus, verifier-feedback aggregation) require additional generation beyond what the present sweep collects, since they need a separate adjudicator output per round. Our framework predicts they lower $\rho^{(0)}$ when the adjudicator is independent of the underlying agents, and leave $\rho^{(0)}$ unchanged when the adjudicator is the same model. Empirical confirmation requires a dedicated experiment.

\section{HumanEval pilot: open-ended generation}
\label{app:humaneval}

\textbf{Setup.} We extend the framework to open-ended code generation as a pilot. Qwen2.5-7B-Instruct on 30 randomly sampled HumanEval problems \citep{chen2021codex} at $N \in \{1, 2, 3, 5, 7, 10\}$, two communication conditions (\emph{comm\_true}: peer code; \emph{comm\_self}: own previous code only), 3 communication rounds, 30 items per cell. We attempted $N{=}30$ at \emph{comm\_true\_k29} but the input from 29 peers' code exceeded the 1024-token output budget on revision rounds; those rows are excluded. Noise placebo (\emph{comm\_random}) was not run.

\textbf{Two measurement axes.} HumanEval admits a discrete \emph{pass-vector} measurement absent in MCQA or numeric math. The problem's \texttt{check} function contains $T$ assertions; an agent's pass-vector is the binary length-$T$ vector indicating which assertions its code passes. We compute pass-vectors deterministically by running each assertion in isolation through \texttt{human\_eval.execution.check\_correctness} (per-assertion sandboxed execution with timeout). Mean $T = 6.9$ across our 30 problems, giving a nominal answer space of ${\sim}2^{6.9}$ per problem. We compute two correlation profiles, both via the variance-matching identity (\S\ref{sec:framework}). On the \emph{pass-vector axis}, agents agree iff their pass-vectors match (functional-outcome axis). On the \emph{code axis}, agents agree iff their canonicalized code matches: each completion runs through Python's AST round-trip (\texttt{ast.parse} $\to$ docstring stripping $\to$ \texttt{ast.unparse}), so cosmetic differences (whitespace, comments, docstring formatting) collapse but algorithmic differences (variable naming, control flow, helper structure) are preserved. This is the open-ended axis.

\textbf{Same task, two regimes (Table~\ref{tab:humaneval}).} The pass-vector axis collapses HumanEval to a hard-ceiling regime ($\beta \in [0.05, 0.08]$, $c \in [0.92, 0.99]$), matching the GPQA pattern. Assertions are highly correlated since most problems are all-pass or all-fail at the function level, so the projection from ${\sim}2^{6.9}$ pass-vectors to the empirical distribution is to one or two modal vectors per problem. The code axis recovers a sublinear regime ($\beta \in [0.12, 0.31]$, $c \in [0.39, 0.59]$), comparable to GSM-Hard. The Ringelmann form fits both axes with $R^2 \ge 0.99$.

\textbf{What this earns.} The same task, measured on a discrete-outcome axis, lands in MCQA-like regime; measured on an open-ended axis, lands in GSM-like regime. The variance-matching identity holds in both, supporting the claim that open-ended generation extends the framework rather than threatening it: $\rho_y$ on a continuous quality measure leaves the form intact. \emph{comm\_self} preserves more code diversity than \emph{comm\_true} ($\beta = 0.31$ vs $\beta = 0.12$ on the code axis): without peer code, agents retain individual algorithmic variation. This contrast is invisible on the pass-vector axis ($\beta_\text{self} = 0.08$ vs $\beta_\text{true} = 0.05$) because both collapse to the same all-pass or all-fail outcome.

\textbf{Caveats.} 30 items per condition is sufficient for $R^2 > 0.99$ fits but marginal for bootstrap CIs (we report point fits only). $N$ is capped at 10. Two conditions only (\emph{comm\_random} / noise placebo not run), so the inert-peer claim from the main paper is not directly replicated here. Cross-model and cross-scale replication on HumanEval are future work. The code-canonicalization metric treats variable-name choices (\texttt{sequence} vs \texttt{tribonacci}) as semantic; an alpha-renaming variant that further normalizes variable identifiers is a natural alternative we have not pursued.

\begin{table}[H]
\centering
\small
\caption{HumanEval pilot: per-$N$ efficiency under two measurement axes (final round, 30 items per cell). $R^{\text{pv}}_N$ uses the pass-vector axis (functional outcome); $R^{\text{code}}_N$ uses the canonicalized-AST code axis (open-ended).}
\label{tab:humaneval}
\begin{tabular}{r rrr rrr}
\toprule
& \multicolumn{3}{c}{comm\_true} & \multicolumn{3}{c}{comm\_self} \\
\cmidrule(lr){2-4} \cmidrule(lr){5-7}
$N$ & $R^{\text{pv}}_N$ & $R^{\text{code}}_N$ & acc & $R^{\text{pv}}_N$ & $R^{\text{code}}_N$ & acc \\
\midrule
1  & \multicolumn{6}{c}{ensemble: $R^{\text{pv}} = R^{\text{code}} = 1.00$, acc $= 0.73$} \\
2  & 0.52 & 0.68 & 0.67 & 0.53 & 0.77 & 0.73 \\
3  & 0.34 & 0.45 & 0.73 & 0.37 & 0.64 & 0.77 \\
5  & 0.22 & 0.34 & 0.73 & 0.23 & 0.52 & 0.73 \\
7  & 0.16 & 0.29 & 0.70 & 0.19 & 0.45 & 0.73 \\
10 & 0.11 & 0.19 & 0.77 & 0.12 & 0.37 & 0.77 \\
\bottomrule
\end{tabular}

\smallskip
\textit{Ringelmann fits.} comm\_true: pass-vec $(c, \beta) = (0.99, 0.05)$, $R^2 = 1.00$; code $(0.59, 0.12)$, $R^2 = 0.99$. comm\_self: pass-vec $(0.92, 0.08)$, $R^2 = 1.00$; code $(0.39, 0.31)$, $R^2 = 1.00$.
\end{table}


\end{document}